\documentclass[makeidx]{amsart}
\usepackage{amsmath}
\usepackage{amssymb}
\usepackage{epsfig}
\usepackage{hyperref}
\newtheorem{Proposition}{Proposition}
  \newtheorem{Remark}{Remark}
  \newtheorem{Corollary}[Proposition]{Corollary}
  \newtheorem{Lemma}[Proposition]{Lemma}
  
  \newtheorem{Theorem}{Theorem}
 
 \newtheorem{Definition}[Proposition]{Definition}

\newcommand{\be}{\begin{equation}}
\newcommand{\ee}{\end{equation}}
\def\blackslug{\hbox{\hskip 1pt \vrule width 4pt height 8pt depth 1.5pt
\hskip 1pt}}
\def\qed{\quad\blackslug\lower 8.5pt\null\par}

\def\Re{\mathrm{Re}}
\def\Im{\mathrm{Im}}

\makeindex

\author{S. Tanveer$^1$} \address{$^1$ Mathematics Department\\The Ohio State University\\Columbus, OH 43210 } 

\title{Analytical approximation for 2-D nonlinear periodic deep Water Waves}

\begin{document}
$ $ \vskip -0.2cm
\begin{abstract}
  A recently developed method \cite{Costinetal}, \cite{Dubrovin}, \cite{Blasius} 
  has been extended to an nonlocal equaton arising in steady water
  wave propagation in two dimensions. We obtain  
  analytic approximation of steady
  water wave solution in two dimensions 
  with rigorous error bounds for a set of parameter values
  that correspond to heights slightly
  smaller than the critical. The wave shapes are shown to
  be analytic. 
  The method presented is quite general and does not
  assume smallness of wave height or steepness and can be 
  readily extended to other interfacial problems involving
  Laplace's equation. 
  \end{abstract}

\vskip -2cm
\maketitle

\today

\section{Introduction}

Recently \cite{Costinetal}, \cite{Dubrovin}, \cite{Blasius}, 
a method has been developed
for study of nonlinear differential equations 
where strong nonlinearity can be
reduced to weakly nonlinear analysis even when
the problem has no natural perturbation parameter.
The idea is quite natural: 
consider an equation in the form $\mathcal{N} [u] =0$, where
$\mathcal{N}$ is some nonlinear operator in some suitable function space. 
A crucial part of this process is to determine
a quasi-solution 
$u_0$ so that $R=\mathcal{N} [u_0]$ is small in an appropriate norm 
and $u_0$ comes close to satisfying appropriate initial and/or boundary conditions.
Then, proving that there exists solution 
$u$ satisfying $\mathcal{N} [u]=0$ is equivalent to showing
that $E=u-u_0$ satisfies appropriately small initial/boundary conditions and
$\mathcal{L} [E] = -R - \mathcal{N}_1 [E]$, 
where the linear operator
$\mathcal{L}$ is the Fre'chet derivative $\mathcal{N}_{u}$ at
$u=u_0$ and $\mathcal{N}_1 [E]= \mathcal{N} [ u_0 + E ] - 
\mathcal{N} [u_0]- \mathcal{L} [ E]$ 
contains only nonlinear terms. 
When $\mathcal{L}$ is suitably invertible subject to initial/boundary conditions
and the nonlinearity $\mathcal{N}_1 $ sufficiently regular, then
standard contraction mapping 
provides a rigorous proof of existence of solution 
to the weakly nonlinear problem for $E$. 
Thus existence of solution to original problem $\mathcal{N} [u]=0$ is
shown, while 
at the same time a rigorous error bound on $u-u_0$ is obtained. An added benefit to this method relative
to abstract nonconstructive methods for proving solutions is that one obtains a concrete
expression for
the approximate solution $u_0$.  
The only non-standard part of this program is to come up with good candidates for quasi-solution
$u_0$. In previous studies 
\cite{Costinetal}, \cite{Dubrovin}, \cite{Blasius}, this has involved
application of classical orthogonal polynomial approximations in
finite domains coupled with exponential asymptotic approach in its complement when
domains extend to $\infty$.

In the present paper, we show that the quasi-solution method can be extended to a 
nonlinear integral equation arising in propagation of steady two dimensional
deep water waves of finite amplitude for a set of values in a range of wave heights.
We provide accurate efficient
representaton for
water waves and at the same time provide rigorous error bounds 
for these approximations.
The literature for water waves is quite extensive and goes back two centuries involving
some of the best known mathematicians Laplace, Langrange, Cauchy, Poisson, Airy, Stokes and
many others (see, for instance, a recent review \cite{Strauss}).
There are many aspects of the water wave problem; these include steady state analysis, 
linear and nonlinear stability of these states, 
the initial value problem and long time behavior. There is also much interest in
finite depth wave propagation and in particular limiting cases when KdV or
Boussinesq models are valid. There is also interest in waves in the presence of shear and other
variants that arise in modeling wind-water interaction. 
The effect of boundaries is also of interest. In principle, the method given here can
be extended to every one of these problems.

Here we are concerned only with steady periodic solutions in two dimensions in deep water.
Existence of steady two dimensional periodic deep water waves 
of small amplitudes was shown by Nekrasov \cite{Nekrasov}, Levi-Civita \cite{Levi-Civita}.
Larger amplitude waves were also studied more recently 
\cite{Krasovskii}, \cite{Keady}, 
\cite{Toland}
culiminating in the proof \cite{Amick-Fraenkel-Toland} 
of Stokes' conjecture of a 
a $120^o$ angle at the apex of the wave with highest height $h_M$.
There have been many numerical calculations as well 
for water waves including an elucidation of the delicate behavior
near highest wave (see for instance \cite{Yamada}-\cite{Longuet-Higgins2}
some of which have been proved \cite{Amick-Fraenkel}, \cite{McLeod} 
Further, there is numerical evidence for bifurcation to 
to periodic waves with multiple crests with unequal heights\cite{Chen-Saffman} as well as to
non-symmetric waves\cite{Zufiria} that is yet to be proved.

It is also interesting to note that
the mathematical formulation used in numerical calculations 
and rigorous analysis have been rather different; one relying
on series representation similar in the spirit of Stokes, 
while the other relies primarily on
integral reformulation due to
to Nekrasov \cite{Nekrasov}.
The present approach is constructive in that we present
approximate solution with
rigorous error bounds; hence proof of existence of solution
follows as a consequence.
In some sense, the approach combines constructive
numerical calculations with mathematical rigor.
We expect this to be helpful both in the rigorous stability 
analysis and bifurcation studies where details of the solution
are likely to be critical.
Another important aspect of the present analysis is that
the approach is quite general and may be readily extended to
other free boundary problems, particularly ones that involve
analytic functions of a complex variable (for {\it e.g.} Hele-Shaw Flow, 
Stokes Bubbles, Vortex patches, just to name a few ).
Further, the rigorous error control method shown here
does not use any special property of the operators in the integral 
formulation of Nekrasov \cite{Nekrasov}.
Instead, with an eye towards generalization to other interfacial problems, 
we employ a straight forward 
series
representation and use spaces isometric to
a weighted $l^1$ space. A bi-product of the analysis is that
analyticity of the boundary follows for waves with a sequence of heights smaller
than the critical for which quasi-solutions have been determined, though
analyticity also follows from other methods in more general contexts \cite{Nicholls}, \cite{BonaLi}.

\section{Steady Water Waves Formulations}

We non-dimensionalize length and time scales implicit 
in setting wavelength and gravity constant $g$ to be $2 \pi$ and $1$ respectively, 
It is known that
the existence of a 
steady symmetric water wave in two dimensions 
when vorticity is unimportant     
is equivalent to showing that there exists analytic
function $f$ 
inside the unit $\zeta$-circle so that $(1+\zeta f^\prime ) \ne 0$ for $|\zeta| \le 1$ and
\be
\label{1}
\Re f = - \frac{c^2}{2 \Big | 1 + \zeta f^\prime \Big |^2} \   ~{\rm on} ~|\zeta|=1 \ ,
\ee 
where $c$ is the non-dimensional wave speed.
Further, for symmetric water waves, $f$ is real valued on the real diameter $(-1, 1)$,
implying real ${\hat f}_j$ in the following representation of $f$:
\be 
\label{2.0}
f(\zeta) = \sum_{j=0}^\infty {\hat f}_j \zeta^j 
\ee
It is to be noted that $i \left ( \log \zeta + f(\zeta) \right ] +2 \pi $
is the conformal map that maps the interior of a cut unit-circle to a periodic strip in
the water-wave domain in a frame where wave profile is stationary, 
with $\zeta=\pm1$ corresponding to 
to wave trough and crest, respectively.
The condition $1+\zeta f^\prime \ne 0$ in $|\zeta| \le 1$ ensures
univalency of the conformal map.    
The formulation (\ref{1})-(\ref{2.0}) is closely related to
those used by others including
Stokes himself. 
Nekrasov\cite{Nekrasov} integral reformulation also follows directly from it as
discussed in the ensuing and involves a parameter
\be
\label{2.0.0.0}
\mu = 
\frac{v_{{\rm crest}}^3}{3 {\tilde c}}  = 
\frac{c^2}{3 \Big | 1 + \zeta f_\zeta \Big |_{\zeta=-1}^3}  
\ee
where $v_{\rm crest} $ is the dimensional speed of fluid at the crest
and ${\tilde c}$ is the dimensional wave speed.
For efficiency in representation,
it is better to represent $f$ in a series in $\eta$: 
\be
\label{3.0} 
f(\eta) = \sum_{j=0}^\infty F_j \eta^j \ , 
\ee
where 
\be
\label{2.1}
\eta = \frac{\zeta+\alpha}{1+\alpha \zeta} \ ,
\ee
for $\alpha \in (0,1)$, where $\alpha$ will be appropriately chosen.
The crest speed parameter (\ref{2.0.0.0}) in this formulation becomes  
\be
\label{2.0.0.0.0}
\mu = 
\frac{c^2}{3 \Big | 1 + \eta q f_\eta \Big |_{\eta=-1}^3}  \ , 
\ee
where
\be
\label{1.2}
q (\eta) = \frac{(\eta-\alpha) (1-\alpha \eta)}{\eta (1-\alpha^2)} 
\ee
The non-dimensional wave height\footnote{Some authors define $2 h$ as the non-dimensional height, while others
present results for the scaled height $\frac{h}{\pi}$}
is given by
\be
\label{2.1.1}
h = -\frac{1}{2} \left [ \Re f ~(1) - \Re ~f (-1) \right ]  = -
\sum_{j=1,odd}^\infty f_j   \ .
\ee
Earlier evidence \cite{Grant}, \cite{Schwartz2}, \cite{Tanveer} 
suggests that for deep water waves 
with one trough and one peak in a period, there is only
one 1/2 singularity of $f$ at $\zeta=-\zeta_s$ for $\zeta_s^{-1} \in (0,1)$ in the
finite complex plane and a fixed logarithmic type singularity at $\zeta =\infty$.
Evidence suggests that $\zeta_s^{-1} $ increases monotonically with
$h \in (0, h_M) $, where $h_M \approx 0.4435\cdots $\footnote{Reported values of
$\frac{h_M}{\pi}$ from computation differ slightly between \cite{Schwartz-Fenton}
and \cite{Williams} ($0.1412$ versus $0.141063$)}
corresponds to the Stokes
highest wave that makes a $120^0$ angle at the apex.
If 
$\mu \in \left ( 0, \frac{1}{3} \right )$ is used as a parameter,
$\mu = \frac{1}{3}$ corresponds to $h=0$, while $\mu=0$ corresponds
to Stokes highest wave $h=h_M$, which has a stagnation point at the crest.
The optimal choice of $\alpha$ that 
ensures the most rapid decay $f_j$ with $j$
is one where
$\zeta=-\zeta_s$, $\zeta=\infty$
are 
mapped to equidistant points from the origin in the $\eta$ plane, {\it i.e.}
when $\alpha=\alpha_0 = \zeta_s - \sqrt{\zeta_s^2 -1}  $. Since the relation of $\zeta_s$ 
with height (or $\mu$) 
is only known numerically, we choose a simple empirical relation:
\be
\label{alphachoice}
\alpha = \frac{2}{237} + \frac{67}{11} \left ( \frac{1}{3} - \mu \right ) 
-\frac{113}{3} 
\left ( \frac{1}{3} - \mu \right )^2 
+\frac{165}{2} 
\left ( \frac{1}{3} - \mu \right )^3 
\ee
that appears to be 
optimal for small $\mu$ corresponding to large amplitude waves; 
the choice is
is not the best for small $\frac{1}{3} -\mu$, but 
it matters little since $f_j$ decays rapidly in any case for
small wave height.
Note that any choice of 
$\alpha \in (0, 1)$ still ensures a convergent series for $f$ in $\eta$; an optimal choice of $\alpha$ 
ensures better accuracy in a finite truncation.
In the $\eta$ variable, the boundary condition (\ref{1}) becomes
\be 
\label{1.1} 
\Re f = -\frac{c^2}{2 \Big | 1 + q (\eta) \eta f^\prime (\eta) \Big |^2}  \mbox  {\rm on} ~|\eta|=1 
\ee
where $q$ is given by (\ref{1.2}).
We note that on the unit $\eta$-circle, $q = \frac{| \eta-\alpha |^2}{1-\alpha^2}$ is real valued.

On $\eta = e^{i \nu}$ and taking derivative with
respect to $\nu$ of the relation (\ref{1.1}) and multiplying through by $q$ (which is real), we obtain
\be
\label{2}
- \Im \left ( q \eta f^\prime \right ) =  
\frac{c^2 q}{\Big | 1 + q \eta f^\prime \Big |^2} \Re \left \{ \frac{d}{d\nu} \log \left ( 1 + \eta q f^\prime \right ) \right \} 
\ee
If we introduce new variable 
\be
\label{3} 
w = - \frac{2}{3} \log c + \log \left ( 1 + \eta q f^\prime \right ) \ , 
{\rm implying} ~\Big | 1 + \eta q f^\prime \Big | = c^{2/3} e^{\Re w}  \ , 
\ee
then (\ref{2}) implies $w$ is analytic in the unit-$\eta$ circle and that on $\eta=e^{i \nu}$, $w$ 
satisfies
\be
\label{4}
\frac{d}{d\nu} \Re w + q^{-1} e^{2 \Re w} \Im e^{w}  = 0  
\ee
This is an alternate formulation of the water wave problem. This is equivalent
to Nekrasov's integral formulation.
If we define $\theta = \Im w$, and integrate (\ref{4}) from $\nu=\pi$ to a variable $\nu$
using the Hilbert transform relation between
$\Re w $ to $\Im w$ on $|\eta|=1$, integration by parts gives the integral equation:
\be
\label{5.0.1}
\theta (\nu) = -\frac{1}{3 \pi} \int_{-\pi}^\pi \log \Big | \sin \frac{\nu-\nu'}{2} \Big | \frac{\sin [\theta (\nu')]}{
q (\nu') \left [ \mu + \int_{\pi}^\nu \frac{\sin \theta (s)}{q(s)} ds \right ]} d\nu' 
\ee
If we set $\alpha=0$ (in which case $q=1$), then (\ref{5.0.1}) reduces\footnote{Rather a change of
variable $\nu \rightarrow \nu-\pi$ gives the Nekrasov form}
to Nekrasov\cite{Nekrasov} integral equation, when oddness of $\phi$ in $\nu-\pi$ is used. In the $w$ variable,
the relation (\ref{2.0.0.0.0}) becomes
\be 
\label{5.0.1.0}
\mu = \frac{1}{3} \exp \left [ - 3 w (-1) \right ] , 
\ee
using $w(-1)$ to be real.

\section{Quasi-solution and transformation to a weakly nonlinear problem}

For given $\mu \in \left (0, \frac{1}{3}\right )$ corresponding to $h \in (0, h_M)$, 
we define a quasi solution $\left ( f_0, c_0 \right ) $ with the property that 
$f_0$ is analytic inside the unit circle
with $1+q\eta f_0^\prime \ne 0$ in
$|\eta| \le 1$, and that on
$\eta=e^{i\nu}$, 
the residual $R_0(\nu)$, defined below, along with its derivative and
the quantity
\be
\label{1.3.wcons}
w_0 (-1) + \frac{1}{3} \log \frac{3}{\mu} = \frac{1}{3} \log \frac{\mu_0}{\mu}
\ee 
are each small enough 
for Proposition \ref{prop2} to hold. Here,
\be
\label{1.3}
R_0(\nu) = \Big | 1 + q (\eta) \eta f_0^\prime (\eta) \Big |^2 \Re f_0 + \frac{c_0^2}{2}  \ ,
\ee
on $\eta=e^{i \nu}$.
We note that if $f_0$ is a polynomial in $\eta$ of order $N$, then $R_0(\nu)$ is a polynomial
in $\cos \nu$ of order $2N+1$, which can be computed\footnote{For $n$ not too large, this can
be done by hand, though use of symbolic language Maple or Mathematica eases the task}
without errors for if $c_0$ and coefficients of the $f_0$ series are chosen as rational numbers.
This can be transformed to a Fourier cosine series with only the first $2N+2$ possibly 
non-zero terms.
 
We note that the representation of the analytic function $w $ inside the unit $\eta$-circle:
\be
\label{4.1}
w (\eta) = \sum_{j=0}^\infty b_j \eta^j \ , ~{\rm where}  ~b_j ~\text{is real} 
\ee
Since $q(\alpha) =0$, it follows that $w (\alpha) = - \frac{2}{3} \log c$, ${\it i.e}$ 
\be
\label{4.2}
-\frac{2}{3} \log c = \sum_{j=0}^\infty b_j \alpha^j
\ee
Corresponding to the quasi-solution $(f_0, c_0)$, we define 
\be
\label{4.3}
w_0 = -\frac{2}{3} \log c_0 + \log \left ( 1 + \eta q(\eta) f_0^\prime \right ) 
\ee   
Then, we can check that 
$w_0$ satisfies
\be 
\label{5}
\frac{d}{d\nu} \Re w_0 + q^{-1} e^{2 \Re w_0} \Im e^{w_0}  = R (\nu)
:= -\frac{R_0^\prime (\nu)}{c_0^2-2 R_0} - \frac{4 A(\nu) R_0 (\nu)}{3 (c_0^2 - 2 R_0)}  
\ ,
\ee
where 
\be
\label{11}
2 A (\nu) = 3 q^{-1} e^{2 \Re w_0} \Im \left \{ e^{w_0} \right \} 
= \frac{3}{c_0^2} \Im \left \{ \eta f_0^\prime \right \} \Big |
1+ q \eta f_0^\prime \Big |^2 \ ,
\ee
It is to be noted that if $f_0$ is a polynomial in $\eta$ of degree $N$, then
(\ref{11}) implies that $\frac{A(\nu)}{\sin \nu} $ is a polynomial in $\cos \nu$ of
order $2N+1$, and therefore $A(\nu)$ has a finite Fourier sine series with only the 
first $2N+2$ terms that are possibly nonzero. Again, as with $R_0$, if $c_0$ and polynomial
coefficients of $f_0$ are given as rationals, the calculation of Fourier sine series
coefficient of $A(\nu)$ can be done without round-off errors. We also note that  
\be
\label{5.1}
w_0 (\alpha) = - \frac{2}{3} \log c_0
\ee
Corresponding to the given quasi-solution $(f_0, c_0)$,
the wave height $h_0$ and wave crest speed parameter $\mu_0$ are given by
\be 
\label{5.1.0}
h_0 = -\frac{1}{2} \left [ f_0 (1) - f_0 (-1) \right ] 
~~\ , ~~ 
\mu_0 = \frac{c_0^2}{3 \Big | 1 + \eta q \partial_\eta f_0 \Big |^3_{\eta=-1}} \ , 
\ee
which may be computed without round-off errors 
for rational $c_0$ and polynomial representation of $f_0$ involving
rational coefficients.

Now, we seek to prove that there are solutions nearby $w_0$. 
For that purpose, we decompose 
\be 
\label{6}
w=w_0+W \ .
\ee
It follows from (\ref{4}) and (\ref{5}) that $W$ satisfies 
\be
\label{7}
\frac{d}{d\nu} \Re W + 2 A(\nu) \Re W + 2 B(\nu) \Im W = \mathcal{{\tilde M}} [W] - R(\nu)
\ee
where on $\eta=e^{i \nu}$,
\be 
\label{11.1}
2 B(\nu) 
= q^{-1} e^{2 \Re w_0} \Re\left \{ e^{w_0} \right \} 
= \frac{1}{q c_0^2} \Big | 1 + \eta q f_0^\prime \Big |^2 
\Re \left [ 1 + q \eta f_0^\prime \right ] \ ,
\ee
and the nonlinear operator $\mathcal{{\tilde M}}$ is defined so that
\be
\label{Mdef}
\mathcal{{\tilde M}} [W] =
-\frac{2}{3}  A(\nu)  M_1 - 2 B (\nu) M_2 \ , 
\ee
where
\be
\label{M1M2def}
M_1 =
e^{2 \Re ~W} ~\Re ~e^{W} - 1 - 3 \Re ~W ~~~\mbox \ , 
~~~M_2 = e^{2 \Re ~W} ~\Im ~e^{W}  - \Im W 
\ee
It is to be noted from (\ref{11.1}) 
that a polynomial $f_0$ in $\eta$ of degree $N$ immediately implies that
${\tilde B} (\nu) = q B(\nu) $ is a polynomial in $\cos \nu$ of degree  
$2 N+2$ and therefore has a truncated Fourier cosine series representation with
at most $2N+3$ terms.
After changes of variable, the constraint (\ref{5.0.1.0}) implies
\be
\label{muconstraint}
W (-1) =  \frac{1}{3} \log \frac{\mu_0}{\mu} \ , {\rm where}      
~\mu_0 = \frac{1}{3} e^{-3 w_0 (-1)} \, 
\ee
which is small from requirement on quasi-solution.
Once a solution is found for $W$, the 
corresponding height of the water wave is given by
\be
\label{7.0}
h = 
h_0 -\frac{1}{2} (1-\alpha^2) 
\int_{-1}^1 \frac{e^{W (\eta) -W (\alpha)} -1 }{(\eta-\alpha) 
(1-\alpha \eta) }
\left [ 1+\eta q(\eta) f_0^\prime (\eta) \right ] d\eta \ ,
\ee
where, noting $f_0$ to be real valued on the real diameter $[-1,1]$,
\be
\label{7.0.0}
h_0 = -\frac{1}{2} \left [ f_0 (1) - f_0 (-1) \right ]
\ee
It is convenient to separate out the linear and nonlinear parts of 
(\ref{7.0}) in the form
\be
\label{7.0.0.0.n}
h=h_0+\mathcal{F} [W] + \mathcal{Q} [W] 
\ee
where the functionals $\mathcal{F}$ and $\mathcal{Q}$ are defined by  
\be
\label{7.0.0.0}
\mathcal{F} [W] =
-\frac{1}{2} (1-\alpha^2) \int_{-1}^1  
\frac{W(\eta) - W (\alpha)}{(\eta-\alpha) (1-\alpha \eta) }
\left [ 1+\eta q(\eta) f_0^\prime (\eta) \right ] d\eta 
\ ,
\ee
\be
\label{7.0.0.0.n.n}
\mathcal{Q} [ W]
= 
-\frac{1}{2} (1-\alpha^2) \int_{-1}^1  
\frac{e^{W (\eta) -W (\alpha)} -1 -W (\eta)+W(\alpha)}{(\eta-\alpha) (1-\alpha \eta) }
\left [ 1+\eta q(\eta) f_0^\prime (\eta) \right ] d\eta \ .
\ee
Once $W$ is determined, the actual wave speed is determined from 
\be 
\label{7.1} 
W (\alpha) = - \frac{2}{3} \log \frac{c}{c_0} 
\ee
Define 
\be
\label{8}
\Phi (\nu) = \Re W (e^{i \nu}) 
\ee
Analyticity of $W $ in the unit circle with sufficient regularity\footnote{The regularity requirements will be
clear in the definition of space $\mathcal{A}$}
upto $|\eta|=1$ implies
\be
\label{9}
\Psi (\nu) = \Im W \left (e^{i \nu} \right ) = 
\frac{1}{2\pi} PV \int_{-\pi}^\pi \cot \frac{\nu-\nu'}{2} \Phi (\nu')  d\nu'
\ee
Then (\ref{7}) may be written abstractly as
\be
\label{9.1}
\mathcal{L} \Phi = \mathcal{M} [\Phi] - R (\nu)  
\ee
where $\mathcal{M} [\Phi] = \mathcal{\tilde M} [ W]$ where $\Phi (\nu) = \Re W (e^{i \nu} ) $
for $W $ analytic in $|\eta| < 1$ and suitably regular in $|\eta| \le 1$, and 
\be
\label{10}
\mathcal{L} \Phi := \Phi^\prime (\nu) + 2 A (\nu) \Phi (\nu) + 2 B (\nu) \Psi (\nu) 
\ee
We will first prove that each given $a_1 \in [-\epsilon_0, \epsilon_0] \equiv I$ for sufficiently
small $\epsilon_0$, 
$\mathcal{L}$ is invertible in an space of functions defined later, and (\ref{9.1})
in that space is equivalent to
\be
\label{eqNon}
\Phi = \mathcal{K} \mathcal{M} [\Phi] - \mathcal{K} R + a_1 G  
:=\mathcal{N} [\Phi ]
\ee
for some function $G $, and $\mathcal{K}$ is a bounded linear operator.
We will then show that for each $a_1 \in I$, the operator $\mathcal{N} $ is contractive in
a small ball in some function space if
quasi-solution satisfies certain conditions that can be readily checked. This corresponds
to a waterwave for which corresponding $\mu$ is in some small neighborhood of $\mu_0$ because
of the relation
\be
\label{crestrelation}
\frac{1}{3} \log \frac{\mu_0}{\mu} =  
W (-1) 
= \Phi (\pi) 
\ee
Using (\ref{eqNon}), we may rewrite (\ref{crestrelation}) in the form
\be
\label{a1relation}
a_1 = \frac{1}{G(\pi)} \left ( \frac{1}{3} \log \frac{\mu_0}{\mu} + \mathcal{K} R [\pi] 
- \mathcal{K} \mathcal{M} [ \Phi] [\pi] \right ) =:\mathcal{U}[a_1] 
\ee
We will then prove $\mathcal{U} : \rightarrow I \rightarrow I$ is contractive
when appropriate smallness conditions are satisfied by quasi-solution and $G(\pi)$ is
not small in which case
there exists unique $a_1 \in I$ 
so that (\ref{crestrelation}) is satisfied for the specified $\mu$.

\section{Definitions, Space of Functions and main results} 

\begin{Definition}
For fixed $\beta \ge 0$, define
$\mathcal{A}$ to be the space of analytic functions 
in $|\eta| < e^{\beta}$ with real Taylor series
coefficient at the origin, equipped with norm: 
\be 
\label{N1}
\| W \|_{\mathcal{A}} = 
\sum_{l=0}^\infty e^{\beta l} \Big | W_{l} \Big | 
\ee
where 
\be
\label{N2}
W (\eta) = \sum_{l=0}^\infty W_l \eta^l \ , 
\ee
\end{Definition}
\begin{Remark}
{\rm 
It is easily seen that $W \in \mathcal{A}$ implies $W$ is continuous in $|\eta| \le e^{\beta}$.
Further, in the domain $|\eta| \le e^{\beta}$, $\| W \|_{\infty} \le \| W \|_{\mathcal{A}} $.
}
\end{Remark}
\begin{Definition}
For $\beta \ge 0$, 
define $\mathcal{E}$ to be the Banach space of real
$2\pi$-periodic even functions $\phi $ 
so
that
\be 
\label{N6}
\phi (\nu) = 
\sum_{j=0}^\infty a_j \cos (j \nu) \ , {\rm with~norm}~
\| \phi \|_{\mathcal{E}}
:=\sum_{j=0}^\infty e^{\beta j} |a_j |  < \infty 
\ee
Define $\mathcal{S}$ to be Banach space of real $2 \pi$- periodic odd functions 
such that
\be
\label{N7}
\psi (\nu) 
=\sum_{j=1}^\infty b_j \sin (j \nu) \ , {\rm with~norm} ~
\| \psi \|_{\mathcal{S}} :=\sum_{j=1}^\infty  e^{\beta j} |b_j |  < \infty 
\ee
\end{Definition}
\begin{Remark}
{\rm 
It is clear that if $\phi \in \mathcal{E}$ if and only if there exists $W \in \mathcal{A}$ so that
$\phi (\nu) = \Re W (e^{i \nu})$.
Similarly, $\psi \in \mathcal{S}$ if and only if
$\psi (\nu) = \Im W (e^{i \nu} )$ for some $W \in \mathcal{A}$. We also note that
for such $W$,
$ \| \phi \|_{\mathcal{E}} = \| W \|_{\mathcal{A}} $, while 
$ \| \psi \|_{\mathcal{S}} \le \| W \|_{\mathcal{A}} $.
}
\end{Remark}

\begin{Remark}
{\rm 
The space $\mathcal{A}$ and $\mathcal{E}$ are clearly isomorphic to each other and to 
$\mathbf{H}$, the space of sequences of real Taylor series coefficients 
${\bf W} = \left (W_0, W_1, \cdots \right )$ with weighted $l^1$ norm
\be
\label{N2.0}
\| {\bf W} \|_{\mathbf{H}} = 
\sum_{l=0}^\infty e^{\beta l} \Big | W_{l} \Big | 
\ee
Because of this isomorphism 
we will move back and forth between spaces $\mathcal{A}$, $\mathcal{E}$ and 
$\mathbf{H}$ as convenient. Similarly the subspace $\mathbf{H}_0 \subset \mathbf{H}$
that consists of all ${\bf W} = \left ( 0, W_1, W_2, \cdots \right )$ is isomorphic to
$\mathcal{S}$. 
}
\end{Remark}

\begin{Theorem}{(Main Result)}
\label{Thm1}
For $\mu \in S $, defined as 
\be
S := \left \{ \mu: \mu \in \cup_{j=1}^{3} \mathcal{I}_{\mu_j} \ , ~{\rm where}~ 
\mu_1 = 0.0018306, \mu_2 = 0.002 \ , \mu_3 = 0.0023
\right \} 
\ee
where $\mathcal{I}_{\mu_j}$ is some sufficiently small interval containing $\mu_j$,
the solution $w$ 
to the water wave problem (\ref{4}) has the representation
\be 
w = -\frac{2}{3} \log c_0 + \log \left (1 + \eta q f_0^\prime \right ) + W
\ee
where quasi-solution $(f_0, c_0)$ is specified in \S \ref{Results} for different cases, and
$W \in \mathcal{A}$ satisfies
error bounds
\be
\| W \|_{\mathcal{A}} \le M_E 
\ee
where $M_E$, depending on $\mu$, is specified in \S \ref{Results} and for all cases
is less than $2.2 \times 10^{-4}$.
The corresponding nondimensional wave speed and heights $(c,h)$ are close to $(c_0, h_0)$
reported in \S \ref{Results} in the sense
that 
\be
|h-h_0| \le K_3 M_E \left ( 1 + 2 e^{1/4} M_E \right )
\ee
\be
\Big | \log \frac{c}{c_0}  \Big | \le \frac{3}{2} M_E 
\ee
for some constant $K_3$ that depends on $\mu$, estimated in \S \ref{Results}.
In all cases considered $K_3 \le 5.24$.
\end{Theorem}
\begin{proof}
The proof of the theorem follows 
by showing that
Propositions \ref{prop2} and \ref{propcrest} 
in the ensuing
apply 
for each proposed quasi-solution  
in \S \ref{Results}
and determing bounds on solutions $\Phi \in \mathcal{E}$ satisfying 
the weakly nonlinear problem (\ref{eqNon}),
where $\Phi = \Re W$.
\end{proof}
\begin{Remark}
{\rm In all likelihood, the error estimates for $M_E$ 
in the Theorem is an over-estimate
by a factor
of about a thousand or so. This is suggested from comparison with a sequence of numerical
calculations with increasing number of modes.}
\end{Remark}

\section{Preliminary Lemmas}

\begin{Lemma}
\label{lemBanachA}
If $W, V \in \mathcal{A}$, $W V\in \mathcal{A}$ and
\be
\| W V \|_{\mathcal{A}} \le \| W \|_\mathcal{A} \| V \|_{\mathcal{A}} 
\ee
\end{Lemma}
\begin{proof}
We note that if
\be 
W (\eta) = \sum_{l=0}^\infty W_l \eta^l \ , V (\eta) = \sum_{l=0}^\infty V_l
\eta^l
\ee 
then using the convolution expression for power series of $W V$,
\be 
\| W V \|_{\mathcal{A}} \le \sum_{k=0}^\infty e^{\beta k} 
\sum_{l=0}^k |V_l| |W_{k-l}| 
= \sum_{l=0}^\infty e^{\beta l} |W_l| 
\left \{ \sum_{j=0}^\infty |W_j| e^{\beta j}  \right \}  
= \| W \|_{\mathcal{A}} \|V \|_{\mathcal{A}}
\ee
\end{proof}
\begin{Corollary}
\label{cor1}
If $W \in \mathcal{A}$, then for any $m \ge 0$, 
\be 
\| \sum_{j=m}^\infty \frac{W^j}{j!}  \|_{\mathcal{A}} 
\le  
\sum_{j=m}^\infty \frac{1}{j!} ~\| W\|_{\mathcal{A}}^j  
= e^{\| W \|_{\mathcal{A}}} - \sum_{j=0}^{m-1} \frac{\| W\|^j}{j!}  
\ee
\end{Corollary}
\begin{proof} The proof follows immediately by using the Banach algebra
property
in Lemma \ref{lemBanachA}
\end{proof}

\begin{Lemma}
\label{lem2}
If $\phi_1, \phi_2 \in \mathcal{E}$, then $\phi_1 \phi_2 \in \mathcal{E}$; if
$\psi_1, \psi_2 \in \mathcal{S}$, then $\psi_1 \psi_2 \in \mathcal{E}$ with
\be
\| \phi_1 \phi_2 \|_{\mathcal{E}} \le \| \phi_1 \|_{\mathcal{E}} 
\| \phi_2 \|_{\mathcal{E}} \,  
\ee
\be
\| \psi_1 \psi_2 \|_{\mathcal{E}} \le \| \psi_1 \|_{\mathcal{S}} 
\| \psi_2 \|_{\mathcal{S}} \ , 
\ee
Further if $\phi \in \mathcal{E}$ and $\psi \in \mathcal{S}$, $\phi \psi 
\in \mathcal{S}$ with 
\be
\| \phi \psi \|_{\mathcal{S}} \le \| \phi \|_{\mathcal{E}} 
\| \psi \|_{\mathcal{S}}  \ ,
\ee
\end{Lemma}
\begin{proof}
Assume 
\be 
\phi_1 (\nu) = \sum_{j=0}^\infty a_j \cos (j \nu) \ , 
\phi_2 (\nu) = \sum_{j=0}^\infty c_j \cos (j \nu) 
\ee
We note that if we define 
${\hat a}_j = \frac{a_{|j|}}{2}$ and   
${\hat b}_j = \frac{b_{|j|}}{2}$ for $j \in \mathbb{Z} \setminus \{0\}$, and
${\hat a}_0 = a_0$, ${\hat b}_0=b_0$,
then $\phi_1$, $\phi_2$ has complex Fourier representations
\be
\phi_1 (\nu) = \sum_{j \in \mathbb{Z}} {\hat a}_j e^{i j \nu} \ ,
\phi_2 (\nu) = \sum_{j \in \mathbb{Z}} {\hat b}_j e^{i j \nu} 
\ee
We also note that in the complex Fourier representation, we may write
\be 
\| \phi_1 \|_{\mathcal{E}} = \sum_{j \in \mathbb{Z}} |{\hat a}_j |e^{\beta |j|} \ , 
\| \phi_2 \|_{\mathcal{E}} = \sum_{j \in \mathbb{Z}} |{\hat b}_j |e^{\beta |j|} 
\ee
Then
\be
\phi_1 (\nu) \phi_2 (\nu) = \sum_{k\in \mathbb{Z}}  e^{i k \nu} 
\sum_{l\in \mathbb{Z}} {\hat a}_{k-l} {\hat b}_{l} 
\ee
Therefore, 
\begin{multline}
\| \phi_2 \phi_2 \|_{\mathcal{E}}
= \sum_{k \in \mathbb{Z}} e^{\beta |k|} \Big |  
\sum_{l\in \mathbb{Z}} {\hat a}_{k-l} {\hat b}_{l} \Big |
\le 
\sum_{k \in \mathbb{Z}}   
\sum_{l\in \mathbb{Z}} e^{\beta |k-l|} 
| {\hat a}_{k-l} | e^{\beta |l|}  | {\hat b}_{l} | 
\\
= 
\left [ \sum_{l \in \mathbb{Z}} | {\hat b}_l | e^{\beta |l|}  
\right ]
\left [ \sum_{j \in \mathbb{Z}} | {\hat a}_j | e^{\beta |j|}  
\right ] = 
\| \phi_1 \|_{\mathcal{E}} 
\| \phi_2 \|_{\mathcal{E}} 
\end{multline}
Assume 
\be 
\psi_1 (\nu) = \sum_{j=1}^\infty a_j \sin (j \nu) \ , 
\psi_2 (\nu) = \sum_{j=1}^\infty c_j \sin (j \nu) 
\ee
We define 
${\hat a}_j = \frac{j}{2 i |j|} a_{|j|}$,
${\hat b}_j = \frac{j}{2 i |j|} b_{|j|}$, 
for $j \in \mathbb{Z}\setminus\{0\}$,
and ${\hat a}_0 =0 ={\hat b}_0$;
then $\psi_1$, $\psi_2$ have complex Fourier representations
\be
\psi_1 (\nu) = \sum_{j \in \mathbb{Z}} {\hat a}_j e^{i j \nu} \ ,
\psi_2 (\nu) = \sum_{j \in \mathbb{Z}} {\hat b}_j e^{i j \nu} . 
\ee
We also note that  
\be
\| \psi_1 \|_{\mathcal{S}} = \sum_{j \in \mathbb{Z}} |{\hat a}_j |e^{\beta |j|} \ , 
\| \psi_2 \|_{\mathcal{S}} = \sum_{j \in \mathbb{Z}} |{\hat b}_j |e^{\beta |j|} 
\ee
Therefore, using the convolution expression in terms of ${\hat a}_j$ and ${\hat b}_j$
it is clear that as for product $\phi_1 \phi_2$,
\be 
\| \psi_1 \psi_2 \|_{\mathcal{E}}  
\le 
\| \psi_1 \|_{\mathcal{S}} \| \psi_2 \|_{\mathcal{S}}  
\ee
The third expression follows in a similar manner using a complex
Fourier Representation.
\end{proof}
\begin{Corollary}
\label{cor2}
If $W \in \mathcal{A}$, then on 
$|\eta| = 1$, for any $m \ge 0$,
\be
\| e^{\Re W (\eta)} - \sum_{j=0}^{m-1} \frac{[\Re ~W (\eta)]^j}{j!} \|_{\mathcal{E}}  
\le 
e^{ \| Re ~W (e^{i \nu}) \|_{\mathcal{E}} } 
- \sum_{j=0}^{m-1} \frac{[\| Re ~W (e^{i \nu}) \|_{\mathcal{E}}^j}{j!} 
\ee
\end{Corollary}
\begin{proof}
This simply follows from noting that
\be
e^{\Re W} - \sum_{j=0}^{m-1} \frac{[\Re ~W]^j}{j!} 
= \sum_{j=m}^\infty \frac{[\Re ~W]^j}{j!} 
\ee
and using the Banach algebra property in the previous Lemma.
\end{proof}
\begin{Lemma}
\label{lemR}
If $R_0 \in \mathcal{E}$, $R_0^\prime \in \mathcal{S}$ and $\| R_0 \|_{\mathcal{E}} < \frac{c_0^2}{2}$,
then $R \in \mathcal{S}$ (recall definition in \eqref{5})  with  
\be
\| R \|_{\mathcal{S}} \le \frac{ \| R_0^\prime \|_{\mathcal{S}} }{c_0^2 - 2 \| R_0 \|_{\mathcal{E}}}  
+ \frac{4 \| A \|_{\mathcal{S}} \| R_0 \|_{\mathcal{E}}}{3 (c_0^2 - 2 \| R_0 \|_{\mathcal{E}}) }
\ee   
\end{Lemma}
\begin{proof}
We use the definition of $R$ in (\ref{5}) and  
Banach Algebra properties in the preceding
lemmas applied to a 
series expansion of $(1-\frac{2}{c_0^2} R_0)^{-1}$ for small $R_0$. The proof readily
follows.
\end{proof}
\begin{Lemma}
\label{lem4}
For $0 \le \beta <\log \alpha^{-1} $, 
if $\phi \in \mathcal{E}$, then 
$ \frac{1}{q} \phi \in \mathcal{E} $, where 
for $\phi = \sum_{l=0}^\infty b_l \cos (l \nu)$,
$\frac{\phi}{q} = \sum_{j=0}^\infty d_j \cos (j \nu) $
where 
\be
d_0 = \sum_{l=0}^\infty b_l \alpha^l \ , ~~d_j
= \alpha^j \sum_{l=0}^{j} b_l \left ( \alpha^{-l} + \alpha^l \right ) 
+ \left ( \alpha^{-j} + \alpha^j \right ) \sum_{l=j+1}^\infty b_l \alpha^l 
~~{\rm for} ~j \ge 1
\ee
and 
\be 
\| q^{-1} \phi \|_{\mathcal{E}} \le C_5 \| \phi \|_{\mathcal{E}} \ , 
\ee
where
\be
C_5 = \frac{2}{1-\alpha e^{\beta}}  
+ \frac{2\alpha}{e^{\beta} - \alpha }
\ee
\end{Lemma}
\begin{proof}
We note that for $j \ge 1$,  
\be
d_j = \frac{1}{\pi} \int_{-\pi}^{\pi} \frac{\phi (\nu)}{q (\nu)} \cos (j \nu) d\nu  
\ee
using $\eta = e^{i \nu}$ on a unit circle counter-clockwise contour integral,
\be 
d_j = \frac{(1-\alpha^2)}{4 \pi i} \sum_{l=0}^\infty b_l \int_{|\eta|=1} 
\frac{d\eta}{(\eta-\alpha)(1-\alpha \eta)} \left ( \eta^j + \eta^{-j} \right )
\left ( \eta^l + \eta^{-l} \right ) 
\ee
On collecting residues, we obtain the expression
\be 
d_j = \alpha^j \sum_{l=0}^{j} b_l \left ( \alpha^{-l} + \alpha^{l} \right )  
+ \left ( \sum_{l=j+1}^\infty b_l \alpha^l \right ) \left (\alpha^{-j} + \alpha^{j} 
\right )  
\ee
Therefore, it follows that
\begin{multline}
\sum_{j=0}^\infty e^{\beta j} |d_j| 
\le 
\sum_{j=0}^\infty \sum_{l=0}^{j} |b_l| e^{\beta l} 
e^{\beta (j-l)} \alpha^{(j-l)}  
\left ( 1 + \alpha^{2l}
\right )  
+\sum_{j=0}^\infty (1+\alpha^{2j}) 
\sum_{l=j+1}^{\infty} |b_l| e^{\beta l} \alpha^{l-j} e^{-\beta (l-j)} \\
\le \frac{2}{1-\alpha e^{\beta}}  
\left ( \sum_{l=0}^\infty |b_l| e^{\beta l} \right ) 
+ 2 \left ( e^{\beta} \alpha^{-1} -1 \right )^{-1} 
\sum_{l=1}^\infty |b_l| e^{\beta l} 
\left ( 1 - e^{-\beta l} \alpha^l \right )
\end{multline}
The same calculation is valid for $j=0$, except for a factor of 2.
\end{proof}
\begin{Remark}
{\rm The above Lemma is very useful in calculating the Fourier cosine coefficients
of $B(\nu)$ defined in (\ref{11.1}) exactly. When $f_0$ is a degree $N$ polynomial in $\eta$, as
mentioned earlier,
$q B(\nu)$ is then a polynomial of $\cos \nu$ of degree $2N+2$, whose coefficients
can be determined without round off errors with rational choice of coefficients.
The above lemma then gives $B_j$ coefficients.}
\end{Remark}
\begin{Lemma}
\label{lem5}
For $0 \le \beta <\log \alpha^{-1} $, 
if $\psi \in \mathcal{S}$, then 
$ \frac{1}{q} \psi \in \mathcal{S} $, where 
for $\psi = \sum_{l=1}^\infty b_l \sin (l \nu)$,
$\frac{\psi}{q} = \sum_{j=1}^\infty d_j \sin (j \nu) $
where 
\be
d_j
= \alpha^j \sum_{l=1}^{j} b_l \left ( \alpha^{-l} - \alpha^l \right ) 
+ \left ( \alpha^{-j} - \alpha^j \right ) \sum_{l=j+1}^\infty b_l \alpha^l 
~~{\rm for} ~j \ge 1
\ee
and 
\be 
\| q^{-1} \psi \|_{\mathcal{S}} \le C_6 \| \phi \|_{\mathcal{S}} \ , 
\ee
where
\be
C_6 = \frac{1}{1-\alpha e^{\beta}}  
+ \frac{\alpha}{e^{\beta} - \alpha }
\ee
\end{Lemma}
\begin{proof}
We note that for $j \ge 1$,  
\be
d_j = \frac{1}{\pi} \int_{-\pi}^{\pi} \frac{\psi (\nu)}{q (\nu)} \sin (j \nu) d\nu  
\ee
using $\eta = e^{i \nu}$ on a unit circle contour integral,
\be 
d_j = -\frac{(1-\alpha^2)}{4 \pi i} \sum_{l=1}^\infty b_l \int_{|\eta|=1} 
\frac{d\eta}{(\eta-\alpha)(1-\alpha \eta)} \left ( \eta^j - \eta^{-j} \right )
\left ( \eta^l - \eta^{-l} \right ) 
\ee
On collecting residues, we obtain the expression
\be 
d_j = \alpha^j \sum_{l=1}^{j-1} b_l \left ( \alpha^{-l} - \alpha^{l} \right )  
+ \left ( \sum_{l=j}^\infty b_l \alpha^l \right ) \left (\alpha^{-j} - \alpha^{j} 
\right )  
\ee
Therefore, it follows that
\begin{multline}
\sum_{j=1}^\infty e^{\beta j} |d_j| 
\le 
\sum_{j=1}^\infty \sum_{l=1}^{j} |b_l| e^{\beta l} 
e^{\beta (j-l)} \alpha^{(j-l)}  
\left ( 1 - \alpha^{2l}
\right )  
+\sum_{j=1}^\infty (1-\alpha^{2j}) 
\sum_{l=j+1}^{\infty} |b_l| e^{\beta l} \alpha^{l-j} e^{-\beta (l-j)} \\
\le \frac{2}{1-\alpha e^{\beta}}  
\left ( \sum_{l=0}^\infty |b_l| e^{\beta l} \right ) 
+ 2 \left ( e^{\beta} \alpha^{-1} -1 \right )^{-1} 
\sum_{l=1}^\infty |b_l| e^{\beta l} 
\left ( 1 - e^{-\beta l} \alpha^l \right )
\end{multline}
\end{proof}

\begin{Lemma}
\label{lemF}
The linear functional $\mathcal{F}$ defined in (\ref{7.0.0.0}) satisfies 
\be
\Big | \mathcal{F} \left [ W \right ] \Big |
\le K_3 \| W \|_{\mathcal{E}} \ ,  
\ee
\be
K_3 = \int_{-1}^1 \frac{ 1+ \eta q f_0^\prime }{1-\alpha \eta} d\eta
\ee
\end{Lemma}
\begin{proof}
Since for $|\eta| \le 1$, 
\be
\Big | \eta^{l-1} + \alpha \eta^{l-2} + 
\alpha^2 \eta^{l-2} \cdots + 
\alpha^{l-1} \Big |
\le \frac{1}{1-\alpha}
\ee
and $1+\eta q f_0^\prime , (1-\alpha \eta) > 0$, it
follows that 
\be
\Big | \mathcal{F} [W] \Big | \le 
\left ( \int_{-1}^1 \frac{1+\eta q_0 f_0^\prime (\eta)}{1-\alpha \eta} d\eta 
\right ) \sum_{l=1}^\infty |W_l |  
\le 
\left ( \int_{-1}^1 \frac{1+\eta q f_0^\prime (\eta)}{1-\alpha \eta} d\eta 
\right ) \| W \|_{\mathcal{A}} =:K_3 \| W \|_{\mathcal{E}}
\ee
\end{proof}
\begin{Remark}
{\rm For polynomial $f_0$, in which case $1+\eta q f_0^\prime$ is also
a polynomial, $K_3$ can be computed exactly as a finite sum of closed form definite
integrals.}
\end{Remark}

\begin{Lemma}
\label{lemQ}
The nonlinear functional $\mathcal{Q}$ defined in
\eqref{7.0.0.0.n.n}
satisfies
the following bounds for $\| W \|_{\mathcal{A}} \le \frac{1}{16}$:
\be
\Big | \mathcal{Q} [ W ] \Big | \le 2 e^{1/4} K_3 \| W \|_{\mathcal{A}}^2  
\ee
\end{Lemma}
\begin{proof}
We note that the functional 
\be
\mathcal{Q} [W] =
\mathcal{F} \left [ e^{U} - 1 - U \right ] \ , {\rm where} ~U(\eta) = W (\eta) - W(\alpha) 
\ee
since $U(\alpha)=0$ and therefore
$e^{U(\alpha)} - 1 - U(\alpha) =0$.
Clearly $U \in \mathcal{A}$ with $\| U \|_{\mathcal{A}} 
\le 2 \| W \|_{\mathcal{A}}$. Applying
Corollary \ref{cor1} and using mean value theorem
and the fact $ 2 \| W \|_{\mathcal{A}} \le 2 B_0 (1+\epsilon) \le \frac{1}{4} $, 
we obtain
\be
\| e^{U} - 1 - U \|_{\mathcal{A}} \le 
e^{\| U \|_{\mathcal{A}}} - 1 - \|U \|_{\mathcal{A}} 
\le 2 e^{1/4} \| W \|_{\mathcal{A}}^2 
\ee
from which it follows that
\be 
\Big | \mathcal{Q} [W] \Big | 
\le 2 e^{1/4} K_3 \| W \|_{\mathcal{A}}^2
\ee  
\end{proof} 

\section{Solving $\mathcal{L} \Phi = r$ for given $a_1 \in I$, $r\in \mathcal{S}$ and bounds on 
$\| \Phi\|_{\mathcal{E}}$}

Consider solving for 
$\Phi\in \mathcal{E}$ satisfying the linear problem $\mathcal{L} \Phi = r$
for given $r \in \mathcal{S}$ and $a_1 \in I$. 
If we use Fourier representation
\be 
\label{13.1}
\Phi (\nu) = \sum_{j=0}^\infty a_j \cos (j \nu) \ ,  
\Psi (\nu) = \sum_{j=1}^\infty a_j \sin (j \nu) 
\ee
\be 
\label{13.2.1}
A(\nu) = \sum_{j=1}^\infty A_j \sin (j \nu) \ ,
B(\nu) = \sum_{j=0}^\infty B_j \cos (j \nu) \ , 
r(\nu) = \sum_{j=1}^\infty r_j \sin (j \nu) 
\ee
Then, equating coefficients of $\sin (k \nu)$ for $k \ge 1$ in the relation
$\mathcal{L} \Phi = r$, where $\mathcal{L}$ given by (\ref{10}), we obtain 
\begin{multline} 
\label{13.2.4}
2 a_0 A_k + \left ( -k + 2 B_0 +A_{2k} - B_{2k} \right ) a_k 
+\sum_{l=1}^{k-1} a_l \left ( A_{k-l} + A_{l+k} + B_{k-l} -B_{l+k} \right )
\\
+ \sum_{l=k+1}^\infty a_l \left ( A_{l+k} - A_{l-k} + B_{l-k} - B_{l+k} 
\right ) = r_k 
\end{multline}
We will solve (\ref{13.2.4}) for $\left ( a_0, a_2, a_3, \cdots. \right )$ for given $a_1 \in I$.
For that purpose, 
it is convenient to re-write (\ref{13.2.4}) in the following form for $k \ge 2$:
\begin{multline} 
\label{13.2.4.0.0}
\frac{2A_k}{l_k} a_0 - a_k
+\sum_{l=2}^{k-1} a_l \frac{1}{l_k} \left ( A_{k-l} + A_{l+k} + B_{k-l} -B_{l+k} \right )
\\
+ \sum_{l=k+1}^\infty a_l \frac{1}{l_k} \left ( A_{l+k} - A_{l-k} + B_{l-k} - B_{l+k} 
\right ) = \frac{r_k}{l_k} 
- \frac{a_1}{l_k} \left ( A_{k-1} + A_{k+1} + B_{k-1} -B_{k+1} \right ) \ ,
\end{multline}
where
\be
\label{14.11.0}
l_k = k - 2 B_0 - A_{2k} + B_{2k} .  
\ee
Quasi solution calculations 
in the range of $h$ reported here show that
$A_1 < 0 $ and $l_k > 0$ for $k \ge 2$; this will be assumed
in the the ensuing.
Setting $k=1$ in
\eqref{13.2.4} leads to
\be
\label{13.2.4.0.0.0}
a_0   
+ \sum_{l=2}^\infty \frac{a_l}{2 A_1} \left ( A_{l+1} - A_{l-1} + B_{l-1} - B_{l+1} 
\right ) = \frac{r_1}{2 A_1} 
+\frac{1}{2 A_1} \left ( 1 - 2 B_0 - A_{2} + B_{2} \right ) a_1 
\ee
Equations \eqref{13.2.4.0.0}
and (\ref{13.2.4.0.0.0}) 
determine a system of equations for 
${\bf a} = \left ( a_0, 0, a_2, a_3 , \cdots \right ) \in \mathbf{H} $ for given 
\be
\label{13.2.4.1.0}
{\bf {\tilde r}} = \Big [ 0, \frac{r_1}{2A_1} + \frac{a_1}{2A_1} (1-2 B_0 - A_2 +B_2 ),  
\left \{ \frac{r_k}{l_k} -
\frac{a_1}{l_k} \left ( A_{k-1} + A_{k+1}+B_{k-1}-B_{k+1} \right )\right \}_{k=2}^\infty 
\Big ]  
\ee
and may be written abstractly as
\be
\label{13.2.4.1}
\mathbf{L} {\bf a} = {\bf {\tilde r}}
\ee
and will consider inversion of $\mathbf{L}$ in the space $\mathbf{H}$ of
sequences as above since this is easily seen to determine solution to $\mathcal{L}\Phi = r$ in the
space $\mathcal{E}$ for given $a_1$. 

\begin{Definition}
We define $\mathbf{H}_0$, $\mathbf{H}_1$ 
to be the subspaces of $\mathbf{H}$ comprising sequences in the form
${\bf a} = (0, a_1, a_2, \cdots )$
and 
${\bf a} = (a_0, 0, a_2, \cdots )$ respectively.
We define $\mathbf{H}_F$ to be 
the (finite) 
$K$-dimensional
subspace of $\mathbf{H}_1$ consisting of all sequences ${\bf a}$ in
the form
\be
{\bf a} = \left ( a_0, 0, a_2, \cdots, a_K, 0, 0, \cdots \right ) 
\ee
Also, we define 
$K$-dimensional subspace $\mathbf{H}_q$ of
$\mathbf{H}_0$ 
consisting of all sequences ${\bf q}$ in the form
\be
{\bf q} = \left (0, q_1, q_2, q_3, \cdots, q_K, 0, 0, \cdots  \right )
\ee 
We define $\mathbf{H}_L$ to be infinite dimensional subspace of $\mathbf{H}$
consisting of all sequences ${\bf a}$ in the form
\be
{\bf a} = \left ( 0, 0, 0, \cdots, 0, a_{K+1}, a_{K+2} \cdots, \right )
\ee
It is clear that $\mathbf{H}_L$ is the compliment of $\mathbf{H}_F$ in $\mathbf{H}_1$, which is the domain of
$\mathbf{L}$, while $\mathbf{H}_L$ is
the compliment of $\mathbf{H}_q$ in $\mathbf{H}_0$, the range of $\mathbf{L}$.
\end{Definition}

It is useful to express ${\bf a} = \left ( {\bf a}_F, {\bf a}_L \right )$, 
${\bf {\tilde r}} = \left ( {\bf {\tilde r}}_q, {\bf {\tilde r}}_L \right )$, where 
${\bf a}_F = \left ( a_0, 0, a_2, \cdots , a_K, 0, 0, \cdots \right ) \in \mathbf{H}_F$, 
${\bf {\tilde r}}_q = \left ( 0, {\tilde r}_1, {\tilde r}_2, \cdots , {\tilde r}_K , 0, 0, \cdots \right )
\in \mathbf{H}_q$,  
${\bf a}_L = \left (0, 0, \cdots 0, a_{K+1}, a_{K+2}, \cdots \right ) \in \mathbf{H}_L$, 
${\bf {\tilde r}}_L = \left (0, 0, \cdots, 0, {\tilde r}_{K+1}, {\tilde r}_{K+2}, \cdots \right ) \in \mathbf{H}_L$.
Then, the system of equation \eqref{13.2.4.1}
may be separated out in the following manner
\be
\label{14.6}
\mathbf{L}_{1,1} {\bf a}_F = - \mathbf{L}_{1,2} {\bf a}_L + {\bf {\tilde r}}_q  \ , 
\mathbf{L}_{2,2} {\bf a}_L = - \mathbf{L}_{2,1} {\bf a}_F + {\bf {\tilde r}}_L  
\ee
where for $k=2, \cdots K$, 
\begin{multline}
\label{14.7}
\left [ \mathbf{L}_{1,1} {\bf a}_F \right ]_k = \frac{2 A_k}{l_k} a_0 - a_k 
+ \sum_{l=2}^{k-1} 
\frac{a_l}{l_k} \left ( A_{k-l} + A_{l+k} + B_{k-l}
- B_{l+k} \right )  \\
+ \sum_{l=k+1}^K \frac{a_l}{l_k} \left ( A_{l+k} - A_{l-k} + B_{l-k} - B_{l+k} \right ) 
\ ,
\end{multline}
while
\be
\label{14.7.0.F}
\left [ \mathbf{L}_{1,1} {\bf a}_F \right ]_0 = 0 ~~\ , ~~
\left [ \mathbf{L}_{1,1} {\bf a}_F \right ]_1 = a_0 + \sum_{l=2}^K \frac{a_l}{2A_1}
\left ( A_{l+1} - A_{l-1} + B_{l-1}
- B_{l+1} \right ) 
\ .
\ee
For $k = 2, \cdots K$,
\be
\label{14.8}
\left [ \mathbf{L}_{1,2} {\bf a}_L \right ]_k = 
\sum_{l=K+1}^\infty \frac{a_l}{l_k} \left ( A_{l+k} - A_{l-k} + B_{l-k} - B_{l+k} \right ) 
\ ,
\ee
while 
\be
\label{14.8.0.F}
[\mathbf{L}_{1,2} {\bf a}_L ]_0 =0 ~~\ , ~~
\left [ \mathbf{L}_{1,2} {\bf a}_L \right ]_1 = \sum_{l=K+1}^\infty \frac{a_l}{2A_1}
\left ( A_{l+1} - A_{l-1} + B_{l-1}
- B_{l+1} \right ) 
\ ,
\ee
and for $k \ge K+1$,
\be
\label{14.9}
\left [ \mathbf{L}_{2,1} {\bf a}_F \right ]_k = \frac{2 A_k}{l_k} a_0 + \sum_{l=2}^{K}
\frac{a_l}{l_k} \left ( A_{k-l} + A_{l+k} + B_{k-l}
- B_{l+k} \right ) \ ,
\ee
\begin{multline}
\label{14.10}
\left [ \mathbf{L}_{2,2} {\bf a}_L \right ]_k = - a_k + \sum_{l=K+1}^{k-1} 
\frac{a_l}{l_k} \left ( A_{k-l} + A_{l+k} + B_{k-l}
- B_{l+k} \right )  \\
+ \sum_{l=k+1}^\infty \frac{a_l}{l_k} \left ( A_{l+k} - A_{l-k} + B_{l-k} - B_{l+k} \right ) 
\ ,
\end{multline}
It is to be noted that $\mathbf{L}_{1,1} : \mathbf{H}_F \rightarrow \mathbf{H}_q$, each being a $K$-dimensional space.
Furthermore, it will be seen that each of 
$\mathbf{L}_{1,2} : \mathbf{H}_L \rightarrow \mathbf{H}_q $, $\mathbf{L}_{2,1} : \mathbf{H}_F \rightarrow 
\mathbf{H}_L $ is a bounded operator. 
We will first show that for sufficient large integer $K$, 
$\mathbf{L}_{2,2}: \mathbf{H}_L \rightarrow \mathbf{H}_L  $.
Then, it will follow from (\ref{14.6}) that ${\bf a}_F $ satisfies
the finite dimensional system of $K$ scalar equations for $K$ unknowns given by
\be
\label{14.7.n.0} 
\left ( \mathbf{L}_{1,1} -\mathbf{L}_{1,2} \mathbf{L}_{2,2}^{-1} 
\mathbf{L}_{2,1} \right ) {\bf a}_F
= {\bf {\tilde r}}_q - \mathbf{L}_{1,2} \mathbf{L}_{2,2}^{-1} {\bf {\tilde r}}_L  
\ee
When
$\mathbf{L}_{1,1}^{-1}$ exists, as may be checked by a finite matrix calculation,
(\ref{14.7.n.0}) implies
\be
\label{14.7.n} 
\left ( I  -\mathbf{L}_{1,1}^{-1} \mathbf{L}_{1,2} \mathbf{L}_{2,2}^{-1} 
\mathbf{L}_{2,1} \right ) {\bf a}_F
= \mathbf{L}_{1,1}^{-1} {\bf {\tilde r}}_q - 
\mathbf{L}_{1,1}^{-1} \mathbf{L}_{1,2} \mathbf{L}_{2,2}^{-1} {\bf {\tilde r}}_L  
\ee
For specific quasi-solution for different $\mu$, we use
explicit matrix computation of $\mathbf{L}_{1,1}^{-1}$ and estimate
$\| \mathbf{L}_{1,1}^{-1} \mathbf{L}_{1,2} 
\mathbf{L}_{2,2}^{-1} \mathbf{L}_{2,1} \|$
in the 
finite dimensional subspace of $\mathbf{H}_{F}$
and demonstrate that it is less than 1, 
implying
\be
\label{14.7.n.n}
\| {\bf a}_F \|_{\mathbf{H}_F} \le 
\left ( 1  -
\| \mathbf{L}_{1,1}^{-1} \mathbf{L}_{1,2} \mathbf{L}_{2,2}^{-1} 
\mathbf{L}_{2,1} \| \right )^{-1}
\left ( \| \mathbf{L}_{1,1}^{-1} \|  \| {\bf {\tilde r}}_q \|_{\mathbf{H}_q}  
+\| \mathbf{L}_{1,1}^{-1} \mathbf{L}_{1,2} 
\mathbf{L}_{2,2}^{-1} \| \| {\bf {\tilde r}}_L \|_{\mathbf{H}_L} \right )  
\ee
Using (\ref{14.6}), we can also estimate $ \| {\bf a}_L \|_{\mathbf{H}_L}$:
\be 
\label{14.7.n.n.n}
\| {\bf a}_L \|_{\mathbf{H}_L} \le 
\| \mathbf{L}_{2,2}^{-1} \mathbf{L}_{2,1} \| \| {\bf a}_F \|_{\mathbf{H}_F}
+ \| \mathbf{L}_{2,2}^{-1} \| \| {\bf {\tilde r}}_L \|_{\mathbf{H}_L}
\ee

\subsection{Bounds on operators}

Consider the system 
\be
\label{14.8.n} 
\mathbf{L}_{2,2} {\bf a}_L = {\bf {\hat r}}_L 
\ee
This is equivalent to the following infinite set of equations for $k \ge K+1$.
\begin{multline} 
\label{13.2.5}
a_k 
=\frac{1}{l_k} \sum_{l=K+1}^{k-1} a_l \left ( A_{k-l} + A_{l+k} + B_{k-l} -B_{l+k} \right )
+\frac{1}{l_k} \sum_{l=k+1}^\infty a_l \left ( A_{l+k} - A_{l-k} + B_{l-k} - B_{l+k} 
\right ) 
- {\hat r}_k \\
=:\left [ \mathbf{M} {\bf a}_L \right ]_k  - {\hat r}_k
\end{multline}
\begin{Lemma}
\label{lemgamma22m}
The operator $\mathbf{M}$ 
defined in (\eqref{13.2.5})
satisfies the following bounds in sub-space $\mathbf{H}_L$:
\be
\label{14.9.m}
\| \mathbf{M} {\bf a}_L \|_{\mathbf{H}_L} \le \gamma \| {\bf a}_L \|_{\mathbf{H}_L}
\ee
where
\begin{multline}
\gamma = 
\sup_{l \ge K+1} \left \{ \sum_{k=l+1}^{\infty} l_k^{-1} 
\Big | e^{\beta (k-l)} \left ( A_{l+k} + A_{k-l}+B_{k-l} - B_{l+k} \right ) \Big | \right .
\\
\left . + 
\left [ \sum_{k=K+1}^{l-1} l_k^{-1} 
\Big | e^{\beta (k-l)} \left ( A_{l+k} - A_{l-k}+B_{l-k} - B_{l+k} \right ) \Big |
\right ] \right \} 
\end{multline}
Further, for $K$ large integer, for ${\bf A}, {\bf B} \in \mathbf{H}$, 
$\gamma $ is small, in which case 
the operator $\mathbf{L}_{2,2}:\mathbf{H}_L \rightarrow \mathbf{H}_L $ 
is invertible with
\be
\| \mathbf{L}_{2,2}^{-1} {\bf {\hat r}} \|_{\mathbf{H}_L} \le 
(1-\gamma)^{-1} 
\|{\bf {\hat r}} \|_{\mathbf{H}_L} =: \gamma_{2,2}^{-1} \| {\bf {\hat r}} \|_{\mathbf{H}_L}
\ee
\end{Lemma}
\begin{proof}
It is convenient to define $m_{k,l}$ so that $m_{k,k}=0$, while
\begin{multline}
m_{k,l} = l_k^{-1} e^{\beta (k-l)} \Big | A_{k-l} + A_{l+k} + B_{k-l} - B_{l+k} \Big | ~{\rm for}~ l < k \\
m_{k,l} = l_k^{-1} e^{\beta (k-l)} \Big | A_{l+k} - A_{l-k} + B_{l-k} - B_{l+k} \Big | ~{\rm for}~ l > k 
\end{multline}
Then, it follows from the definition of $\mathbf{M}$ that 
\be 
\sum_{k=K+1}^\infty e^{\beta k} \Big | \left [ \mathbf{M} {\bf a}_L \right ]_k \Big | 
\le \sum_{k=K+1}^\infty \sum_{l=K+1}^\infty e^{\beta l} |a_l| m_{k,l} 
\le \left \{ \sup_{l \ge K+1} \sum_{k=K+1}^\infty m_{k,l} \right \} \| 
{\bf a}_L \|_{\mathbf{H}_L}  
\ee
from which the first part of the Lemma follows using definition of $m_{k,l}$. 
It is also clear that since for sufficiently large $K$, $l_k$ is an increasing function of $k$ for $k \ge K+1$, 
\begin{multline}
\sum_{k=K+1}^{l-1}  
l_k^{-1} e^{\beta (k-l)} \Big | A_{l+k} - A_{l-k} + B_{l-k} - B_{l+k} \Big | 
+\sum_{k=l+1}^\infty   
l_k^{-1} e^{\beta (k-l)} \Big | A_{k-l} + A_{l+k} + B_{k-l} - B_{l+k} \Big | 
\\
\le 
l_{K+1}^{-1} \left [  
 \sum_{m=1}^{\infty} e^{\beta m} \left ( |A_m| + |B_m| \right )   
+ e^{-2 \beta l} \sum_{m=K+1+l}^{2l-1} e^{\beta m} \left ( |A_m| + |B_m | \right ) \right .\\ 
\left. + e^{-2 \beta l} \sum_{m=2l+1}^{\infty} e^{\beta m} \left ( |A_m| + |B_m | \right ) 
+ \sum_{m=1}^{l-K-1} e^{-\beta m} \left ( |A_m| + |B_m| \right ) \right ] 
\end{multline}
The supremum of the above expression over all $l \ge K+1$ clearly shrinks to 0 as $K \rightarrow \infty$.
Therefore $\gamma$ which is bounded by the above is small for large $K$.
The second part of the Lemma follows readily from
bounds on $\mathbf{M}$. 
\end{proof}

\begin{Lemma}
\label{lemgamma12}
The operator $\mathbf{L}_{1,2} : \mathbf{H}_L \rightarrow \mathbf{H}_q $ 
is bounded and satisfies 
the uniform bound
\be
\label{14.11}
\| \mathbf{L}_{1,2} {\bf a}_L \|_{\mathbf{H}_q}
\le \gamma_{1,2} \| {\bf a}_L \|_{\mathbf{H}_L} , 
\ee
where
\begin{multline}
\label{14.11.0}
\gamma_{1,2} = \sup_{l \ge K+1} 
\sum_{k=2}^K \frac{1}{l_k} \Big | e^{\beta (k-l)} \left ( A_{l+k} 
- A_{l-k} + B_{l-k} - B_{l+k} \right )  \Big | \\
+
\frac{e^{\beta}}{2 |A_1|} 
\sup_{l \ge K+1} e^{-\beta l} \Big | A_{l+1}-A_{l-1}+B_{l-1}-B_{l+1} \Big |
\end{multline}
Furthermore, for large $K$, $\gamma_{1,2}$ is small.
\end{Lemma}
\begin{proof}
From definition of 
$\mathbf{L}_{1,2} $ in (\ref{14.8}), it follows that
\begin{multline}
\label{14.11.1}
\| \mathbf{L}_{1,2}  {\bf a}_{L} \|_{\mathbf{H}_q}
= 
\Big | e^{\beta} \sum_{l=K+1}^\infty  \frac{a_l}{2 A_1} 
\left ( A_{l+1}-A_{l-1}+B_{l-1}-B_{l+1} \right )
\Big | \\
+
\sum_{k=2}^K e^{\beta k} \Big | \sum_{l=K+1}^\infty \frac{a_l}{l_k} \left ( A_{l+k} 
- A_{l-k} + B_{l-k} - B_{l+k} \right )  \Big |
\\
\le 
\left \{ \sum_{l=K+1}^\infty |a_l| e^{\beta l} \right \} 
\left \{ \frac{e^{\beta}}{2 |A_1|} 
\sup_{l \ge K+1} e^{-\beta l} \Big | A_{l+1}-A_{l-1}+B_{l-1}-B_{l+1} \Big | \right .\\
\left .+ \sup_{l \ge K+1} 
\sum_{k=2}^K \frac{1}{l_k} \Big | e^{\beta (k-l)} \left ( A_{l+k} 
- A_{l-k} + B_{l-k} - B_{l+k} \right )  \Big | \right \} \ ,
\end{multline}
from which the first part of the Lemma follows.
We note that since ${\bf A}, {\bf B} \in {\bf H}$, it is clear that 
\be
\Big | A_{l+1} - A_{l-1} + B_{l-1} - B_{l+1} \Big | \rightarrow 0 ~{\rm when} ~l \ge K+1 \, {\rm and}~K
\rightarrow \infty
\ee
Also, we note that
\begin{multline} 
\sum_{k=2}^K \frac{e^{\beta (k-l)}}{l_k} \Big | B_{l-k} - A_{l-k} \Big | 
= \sum_{k=2}^{K/2} 
\frac{e^{\beta (k-l)}}{l_k} \Big | B_{l-k} - A_{l-k} \Big | 
+ \sum_{k=K/2+1}^{K} 
\frac{e^{\beta (k-l)}}{l_k} \Big | B_{l-k} - A_{l-k} \Big | \\
\le \frac{1}{l_2} \sum_{l-K/2}^{l-2} e^{-\beta m} \Big |B_m - A_m \Big |
+\frac{1}{l_{K/2+1}} \sum_{l-K}^{l-K/2-1} e^{-\beta m} \Big |B_m - A_m \Big |
\end{multline}
Clearly, the above shrinks to zero as $K \rightarrow \infty$ for any $l \ge K+1$. Also,
\be
\sum_{k=2}^K e^{\beta (k-l)} \Big | A_{l+k} - B_{l+k} \Big |  
= e^{-2 \beta l} \sum_{m=l+2}^{K+l} e^{\beta m} \Big | A_m - B_m \Big |
\ee
This also shrinks to zero as $K \rightarrow \infty$ for any $l \ge K+1$. 
Therefore, it follows from expression for $\gamma_{1,2} $ that it shrinks to zero as $K \rightarrow \infty$.
\end{proof}
\begin{Lemma}
\label{lemgamma21}
The operator $\mathbf{L}_{2,1}: \mathbf{H}_F \rightarrow \mathbf{H}_L $ is bounded and satisfies
\be
\label{14.12}
\| \mathbf{L}_{2,1} {\bf a}_F \|_{\mathbf{H}_L} \le \gamma_{2,1} \| {\bf a}_F \|_{\mathbf{H}_F} \ , 
\ee
where
\be
\label{14.12.0}
\gamma_{2,1} = \max \left \{ 
\sum_{k=K+1}^\infty \frac{2|A_k|}{l_k} e^{\beta k} \ , 
\sup_{2\le l\le K} 
\sum_{k=K+1}^\infty \Big | \frac{e^{\beta (k-l)}}{l_k} 
\left ( A_{k-l} + A_{l+k} 
+ B_{k-l} - B_{l+k} \right ) \Big | 
\right \}
\ee
Furthermore for large $K$, $\gamma_{2,1} $ is small.
\end{Lemma}
\begin{proof}
Using (\ref{14.9}), we obtain
\begin{multline}
\| \mathbf{L}_{2,1} {\bf a}_F \|_{\mathbf{H}_L} 
\le |2 a_0|  \sum_{k=K+1}^\infty \frac{|A_k|}{l_k} e^{\beta k} 
+\sum_{k=K+1}^\infty \frac{e^{\beta k}}{l_k} \sum_{l=2}^K |a_l| \Big | A_{k-l} + A_{l+k} 
+ B_{k-l} - B_{l+k} \Big | \\
\le 
|a_0| 
\sum_{k=K+1}^\infty \frac{2 |A_k|}{l_k} e^{\beta k} 
+\left [ \sum_{l=2}^K |a_l| e^{\beta l} \right ]
\sup_{2\le l\le K} 
\sum_{k=K+1}^\infty \Big | \frac{e^{\beta (k-l)}}{l_k} 
\left ( A_{k-l} + A_{l+k} 
+ B_{k-l} - B_{l+k} \right ) \Big | \\
\le \gamma_{2,1} \left \{|a_0| 
+ \sum_{k=2}^K e^{\beta k} |a_k| \right \} \ ,
\end{multline}
from which the first part of the lemma follows. We also note that
for sufficiently large $K$,
\begin{multline} 
\sum_{k=K+1}^\infty \frac{e^{\beta (l-k)}}{l_k} \Big | A_{k-l}+B_{k-l} + A_{l+k}-B_{l+k} \Big | 
\le \frac{1}{l_{K+1}} \sum_{m=1}^\infty e^{\beta m} \Big | A_m + B_m \Big |  \\
+ \frac{e^{-2 \beta l}}{l_{K+1}} \sum_{m=K+1-l}^\infty e^{\beta m} \Big | A_m - B_m \Big |  
\ ,
\end{multline}
which shrinks to zero as $K \rightarrow \infty$ for any $ 2 \le l \le K$.  
Furthermore,  
\be
\sum_{k=K+1}^\infty \frac{2}{l_k} |A_k| e^{\beta k} 
\le \frac{2}{l_{K+1}} \sum_{K+1}^\infty |A_k| e^{\beta k} \rightarrow 0
~{\rm as} ~K \rightarrow \infty  
\ee
Therefore $\gamma_{2,1} $ is small for large $K$.
\end{proof}

\begin{Lemma}
\label{lemgamma11}
$\mathbf{L}_{1,1}: \mathbf{H}_F \rightarrow \mathbf{H}_q $ is invertible if and only if the $K \times K$  matrix
${\bf J} = \left \{ J_{k,l} \right \}_{k,l} $ with elements determined by:
\begin{equation}
\label{17.3.J}
J_{1,l} = \frac{e^{\beta (1-l)}}{2 A_1} \left ( A_{l+1} - A_{l-1} + B_{l-1} - B_{l+1} \right )  
~~{\rm for} ~ l=2, 3, \cdots, K \ ,
\end{equation}
$J_{1,1}=e^{\beta}$, $J_{k,k} = -1$,  $J_{k,1} =\frac{2}{l_k} e^{\beta k} 
A_k $ 
for $2 \le k \le K$ 
\begin{equation}
\label{17.3.J.n}
J_{k,l} = \frac{e^{\beta (k-l)}}{l_k} \left ( A_{k-l} + A_{k+l} + B_{k-l} - B_{k+l} \right )  
~~{\rm for} ~ 2 \le l \le k-1 \le K-1
\end{equation}
\begin{equation}
\label{17.3.J.n.n}
J_{k,l} = \frac{e^{\beta (k-l)}}{l_k} \left ( A_{l+k} - A_{l-k} + B_{l-k} - B_{l+k} \right )  
~~{\rm for} ~ 2 \le k \le l-1 \le K-1
\end{equation}
Further $\gamma_{1,1}^{-1} := \| \mathbf{L}_{1,1}^{-1} \| = \|  {\bf J}^{-1} \|_{1} $ where
$\| . \|_1$ denotes the matrix $1$-norm.
\end{Lemma}
\begin{proof}
The proof follows from examining the
the definition of $\mathbf{L}_{1,1}$ in (\ref{14.7})-(\ref{14.7.0.F}) and noting that both the
domain and range of $\mathbf{L}_{1,1}$ is $K$ dimensional.   
The factors of $e^{\beta}$ 
in the matrix elements of ${\bf J}$ ensure that
the $\mathbf{H}$ norm of $\mathbf{L}_{1,1}^{-1}$ is the $1$-norm of the matrix ${\bf J}^{-1}$, if and when
it exists.
\end{proof}

\begin{Proposition}
\label{prop1}
If for some suitably large $K$, $\mathbf{L}_{1,1}^{-1}$ exists with
$\| \mathbf{L}_{1,1}^{-1} \| = \gamma_{1,1}^{-1}$ satisfying 
\be
\label{eqPropcond}
\gamma_{1,2} \gamma_{2,1} \gamma_{2,2}^{-1} \gamma_{1,1}^{-1} < 1
\ee 
then 
$\mathbf{L}^{-1} : \mathbf{H}_0 \rightarrow \mathbf{H}_1$
exists and satisfies
\be
\| \mathbf{L}^{-1} \| \le {\tilde M} , 
\ee
where 
\begin{multline}
{\tilde M} = \max \left \{ 
\gamma_{1,1}^{-1} 
\left ( 1-\frac{\gamma_{1,2} \gamma_{2,1}}{\gamma_{1,1} \gamma_{2,2}} 
\right )^{-1} \left ( 1 + \gamma_{2,2}^{-1} \gamma_{2,1} \right )
, \right . \\   
\left . \gamma_{2,2}^{-1} +\gamma_{1,1}^{-1} 
\left ( 1-\frac{\gamma_{1,2} \gamma_{2,1}}{\gamma_{1,1} \gamma_{2,2}} 
\right )^{-1} \left ( \gamma_{1,2} \gamma_{2,2}^{-1} + 
\gamma_{2,2}^{-2} \gamma_{2,1} \gamma_{1,2} \right ) \right \}
\end{multline}
When condition (\ref{eqPropcond}) is satisfied, for given $a_1 \in \mathbb{R}$, $r \in \mathcal{S}$, 
the linear system
$\mathcal{L} \Phi = r $ has a unique solution in the form
\be
\Phi (\nu) = \mathcal{K} [ r ] (\nu)+ a_1 G (\nu)  \ ,
\ee
where  $\mathcal{K} : \mathcal{S} \rightarrow \mathcal{E}$ is a linear operator
\be
\| \mathcal{K} \| \le {\tilde M} \max   
\left \{ \frac{1}{2 |A_1|}, \sup_{k \ge 2} \frac{1}{l_k} 
\right \} =: M
\ee 
and 
\be 
G(\nu) = g_0 + \cos \nu + \sum_{k=2}^\infty g_k \cos (k \nu)  \ ,
\ee 
where ${\bf g} = \left ( g_0, 0, g_2 , \cdots \right ) \in \mathbf{H}_1 $ is given by
$ {\bf g} = \mathbf{L}^{-1} {\bf h} $
where ${\bf h} = (0, h_1, h_2 , \cdots  ) \in \mathbf{H}_0$, where
\be 
h_1 = \frac{1}{2 A_1} \left ( 1-2 B_0 - A_2 +B_{2} \right ) 
\ee
and for $k \ge 2$,
\be 
h_k = -\frac{1}{l_k} \left ( A_{k-1} + A_{k+1} + B_{k-1} - B_{k+1} \right )  
\ee
Furthermore,
\be
\| G \|_{\mathcal{E}} \le e^{\beta} + {\tilde M} \| {\bf h} \|_{\mathcal{H}_0}
\ee
\end{Proposition}
\begin{proof} The first part follows from applying
estimates in Lemmas \ref{lemgamma22m}-\ref{lemgamma21} to (\ref{14.6}) and
using 
$\| {\bf {\tilde r}} \|_{\mathbf{H}_0}=
\| {\bf {\tilde r}}_q \|_{\mathbf{H}_q}
+\| {\bf {\tilde r}}_L \|_{\mathbf{H}_L}
$,
$\| {\bf a}\|_{\mathbf{H}_1}=
\| {\bf a}_F \|_{\mathbf{H}_F}
+\| {\bf a}_L \|_{\mathbf{H}_L}$.
For the second part, we note that if ${\bf r} =0$, then ${\bf {\tilde r}}={\bf h} a_1$ and
therefore in that case ${\bf a} = a_1 \mathbf{L}^{-1} {\bf h} \in \mathbf{H}_1$. ${\bf a}$
is isomorphic to $a_1 \Phi \in \mathcal{E}$ where
$\mathcal{L} \Phi = - \mathcal{L} [\cos \nu ]$, where $\Phi$ has the form
$a_0 + \sum_{l=2}^\infty a_l \cos (l \nu)$. Therefore, $G= \Phi + \cos \nu$ 
is the unique solution to $\mathcal{L} [G]=0$ with unit coefficient of $\cos \nu$.
When $a_1=0$, but ${\bf r} \ne 0$, 
${\bf {\tilde r}} = \left ( \frac{r_1}{2 A_1}, \left \{ \frac{r_k}{l_k} \right \}_{k=2}^\infty\right ) $ 
for which case
\be
\label{eqfac}
\| {\bf {\tilde r}} \|_{\mathbf{H}_0} \le \max 
\left \{ \frac{1}{2 |A_1|}, \sup_{k \ge 2} \frac{1}{|l_k|} 
\right \} \| {\bf r} \|_{\mathbf{H}_0}
\ee
and corresponding ${\bf a} \in \mathbf{H}_1 $ is isomorphic to 
$\Phi \in \mathcal{E}$ with no $\cos \nu$ term uniquely satisfying $\mathcal{L} \Phi = r$.
This is defined to be $\mathcal{K} r $.
Using linear superposition of the two cases: {\bf i}. $a_1 \ne 0$, ${\bf r} =0$ and {\bf ii.} $a_1=0$, 
${\bf r}\ne 0$ gives the 
the second part of the proposition. The bounds on $\| G\|_{\mathcal{E}}$ follow from
the bounds on ${\bf L}^{-1} {\bf h}$ and adding to it the contribution from the $\cos \nu$ term. 
\end{proof}

\begin{Corollary}
\label{cor3}
For $a_1 \in I:=[-\epsilon_0, \epsilon_0] $, define
$\Phi^{(0)} (\nu) = -\mathcal{K} R + a_1 G = \mathcal{N} [ 0 ]$, where operator $\mathcal{N}$ is
defined in (\ref{eqNon}).
Then $\Phi^{(0)} $ satisfies
\be
\label{B0def}
\| \Phi^{(0)} \|_{\mathcal{E}} \le M \| R \|_{\mathcal{S}} + \epsilon_0  
\| G \|_{\mathcal{E}} =: B_0     
\ee
\end{Corollary}
\begin{proof}
The proof follows from bounds on operator $\mathcal{K}$ in the previous proposition. 
\end{proof}
\begin{Lemma} 
\label{lemG0}
Assume $G_0\in\mathcal{E}$ is an approximate expression for $G$ in the sense that
$\| \mathcal{L} G_0 \|_{\mathcal{S}} = \epsilon_G $ is small and $\cos \nu$ coefficient of
$G_0$ is also 1, as for $G$. If conditions of Proposition \ref{prop1} hold, then
\be
\| G - G_0 \|_{\mathcal{E}} \le M \epsilon_G 
\ee  
In particular, 
\be
\| G \|_{\mathcal{E}} \le \| G_0 \|_{\mathcal{E}} + M \epsilon_G 
\ee
and
if $G_0(\pi) \ne 0$ and $\epsilon_G$ is sufficiently small then
\be 
\Big |G (\pi) \Big | > \Big | G_0 (\pi ) \Big | 
- M \epsilon_G > 0 
\ee
\end{Lemma}
\begin{proof}
Since $\mathcal{L} [G-G_0] = -\mathcal{L} [G_0] $ and coefficient of $\cos \nu$ for $G-G_0$ is zero,
applying Proposition \ref{prop1}, it follows that
\be
G-G_0 = -\mathcal{K} \mathcal{L} G_0     
\ee
which gives the result 
\be
\| G - G_0 \|_{\mathcal{E}} \le M \epsilon_G 
\ee
The remaining two parts of the Lemma follow from triangular inequality and 
the observation
$\Big | G(\pi)-G_0 (\pi) \Big | 
\le \| G - G_0 \|_{\mathcal{E}}$.
\end{proof}

\section{Nonlinearity bounds and solution to (\ref{eqNon}) for given $a_1\in I$}

\begin{Proposition}
\label{prop1.1}
$\mathcal{{\tilde M}}$ defined in (\ref{Mdef}) 
satisfies $\mathcal{{\tilde M}}:\mathcal{A} \rightarrow 
\mathcal{S}$ with
\begin{multline}
\| \mathcal{{\tilde M}} [ W ] \|_{\mathcal{S}} \le 
\frac{2}{3} \| A \|_{\mathcal{S}} \left [ 
e^{2 \| W \|_{\mathcal{A}} } 
\left ( e^{\| W \|_{\mathcal{A}}} - 1 - \| W \|_{\mathcal{A}} \right ) \right. \\
\left . + \left ( 1 + \| W \|_{\mathcal{A}} \right )  
\left ( e^{2 \| W \|_{\mathcal{A}}} - 1 - 2 \| W \|_{\mathcal{A}} \right ) 
+2 \| W \|^2_{\mathcal{A}} \right ]
\\
+2 \| B \|_{\mathcal{E}} \left [
e^{2 \| W \|_{\mathcal{A}}} \left ( e^{\| W \|_{\mathcal{A}}}
-1 - \| W \|_{\mathcal{A}} \right ) 
+ \| W \|_{\mathcal{A}} \left ( e^{2 \| W \|_{\mathcal{A}}} - 1 \right ) 
\right ]
\end{multline}
In particular if $\|W \|_{\mathcal{A}} \le \frac{1}{16}$,
\be 
\| \mathcal{{\tilde M}} [W ] \| \le \left ( 4 \| A \|_{\mathcal{S}} + 
6 \| \mathcal{B} \|_{\mathcal{E}} \right ) \| W \|_{\mathcal{A}}^2
\ee
\end{Proposition}
\begin{proof}
Recall $M_1$ and $M_2$ defined in expression \eqref{Mdef} defining $\mathcal{\tilde M}$.
We may rewrite 
\be 
M_1 = e^{2 \Re W} \Re \left ( e^{W} -1 - W \right ) 
+ \Re ~(1+W) ~\left ( e^{2 \Re W} - 1 - 2 \Re W \right ) 
+ 2 [\Re W ]^2 
\ee
Using Corollary \ref{cor2},
\be
\| M_1 \|_E \le 
e^{2 \| W \|_{\mathcal{A}} } 
\left ( e^{\| W \|_{\mathcal{A}}} - 1 - \| W \|_{\mathcal{A}} \right ) 
+ \left ( 1 + \| W \|_{\mathcal{A}} \right )  
\left ( e^{2 \| W \|_{\mathcal{A}}} - 1 - 2 \| W \|_{\mathcal{A}} \right ) 
+ 2 \| W \|_{\mathcal{A}}^2
\ee
Also, we have from (\ref{Mdef}), we may write
\be 
M_2 = e^{2 \Re W} \Im \left ( e^{W} -1 - W \right ) 
+ \Im ~W ~\left ( e^{2 \Re W} - 1 \right ) 
\ee
Therefore, using corollaries \ref{cor1} and \ref{cor2},
\be
\| M_2 \|_{\mathcal{S}} \le 
e^{2 \| W \|_{\mathcal{A}}} \left ( e^{\| W \|_{\mathcal{A}}}
-1 - \| W \|_{\mathcal{A}} \right ) 
+ \| W \|_{\mathcal{A}} \left ( e^{\| 2 W \|_{\mathcal{A}}} - 1 \right ) 
\ee
Therefore from Lemma \ref{lem2}, $\mathcal{\tilde M} [ W ] \in \mathcal{S} $ and  
\be
\| \mathcal{\tilde M} [ W ] \|_{\mathcal{S}}  
\le \frac{2}{3} \| A \|_{\mathcal{S}} \| M_1 \|_{\mathcal{E}}
+2 \| B \|_{\mathcal{E}} \| M_2 \|_{\mathcal{S}}
\ee
from which the first part of the Lemma follows. The second statement
can be checked by use of mean value theorem to estimate
$e^{z}-1-z$ and $e^{z}-1$.
\end{proof}

\begin{Proposition}
\label{prop2}
For given $a_1 \in I$, $\mathcal{N}$, defined in
\eqref{eqNon}, satisfies
$\mathcal{N}: \mathcal{B} \rightarrow \mathcal{B} $
and is contractive in the  
ball $\mathcal{B}\subset\mathcal{E}$ 
of radius $B_0 (1+\epsilon)$  about the origin ($B_0$ defined in (\ref{B0def}))
if there exists
$\epsilon > 0$ so that $B_0 (1+\epsilon) \le \frac{1}{16}$ and
the following conditions hold:
\be 
\label{eqCont}
M \left ( 4 \| A \|_{\mathcal{S}} + 6 \| B \|_{\mathcal{E}} \right ) B_0 (1+\epsilon)^2 < \epsilon
 \ ,
2 M \left ( 4 \| A \|_{\mathcal{S}} + 6 \| B \|_{\mathcal{E}} \right ) B_0 (1+\epsilon) < 1 
\ee
When these conditions are satisfied,
(\ref{eqNon}) has unique solution $\Phi \in \mathcal{B} \subset \mathcal{E}$.
Each such choice of $a_1$ corresponds 
to a symmetric water wave with nondimensional height, wave speed and crest speed 
$(h,c, \mu)$ close to $(h_0, c_0, \mu_0)$
satisfying the following estimates:
\be
|h-h_0| \le K_3 (1+\epsilon) B_0 \left ( 1 + 2 e^{1/4} B_0 (1+\epsilon) \right )
\ee
\be
\Big | \log \frac{c}{c_0}  \Big | \le \frac{3}{2} B_0 (1+\epsilon)
\ee
\be
\frac{1}{3} \Big | \log \frac{\mu_0}{\mu} \Big | \le B_0 (1+\epsilon)  
\ee
\end{Proposition}
\begin{proof}
Applying Propositions \ref{prop1} and \ref{prop1.1} to \eqref{eqNon} for
for $\Phi^{(1)} , \Phi^{(2)} \in \mathcal{B}$,
\begin{multline}
\| \mathcal{N} [ \Phi^{(1)} ] - \mathcal{N} [ \Phi^{(2)} \|_{\mathcal{E}}
= \| \mathcal{K} \mathcal{M} [\Phi^{(1)} ] 
-\mathcal{K} \mathcal{M} [\Phi^{(2)} ] \|_{\mathcal{E}} \\
\le M \left ( 4 \| A \|_{\mathcal{S}} + 6 \| B \|_{\mathcal{E}} \right ) \| 
\left \| \left ( \Phi^{(1)} + \Phi^{(2)} \right ) 
\left ( \Phi^{(1)} - \Phi^{(2)} \right ) \right \|_\mathcal{E}     
\end{multline}
Using this and given condition (\ref{eqCont})
\begin{multline}
\| \mathcal{N} \left [ \Phi \right ] 
\|_{\mathcal{E}} \le 
\| \mathcal{N} \left [ 0 \right ]\|_{\mathcal{E}} 
+\| \mathcal{N} \left [ \Phi \right ] -\mathcal{N} [0 ]
\|_{\mathcal{E}} \\
\le B_0    
+ M \left ( 4 \| A \|_{\mathcal{S}} + 6 \| B \| \right ) \| 
B_0^2 (1+\epsilon)^2 \le B_0 (1+\epsilon)
\end{multline}
Therefore $\mathcal{N} : \mathcal{B} \rightarrow \mathcal{B}$ contractively and 
the integral equation
$\Phi = \mathcal{N} [\Phi]$ has a unique solution in $\mathcal{B}$.
The estimates on $\log \frac{c}{c_0}$, $\log \frac{\mu_0}{\mu} $ and $h-h_0$ follow from 
(\ref{7.1}), (\ref{crestrelation}) and
(\ref{7.0.0.0.n}) by applying estimates in Lemmas \ref{lemF}, \ref{lemQ} and the bound 
$\| W\|_{\mathcal{A}} \le B_0 (1+\epsilon)$.  
\end{proof}

\begin{Remark} 
{\rm When 
residual size $R$ and interval $I$ are each sufficiently small,
Proposition \ref{prop2} gives solution $w=w_0+W$ 
to the water wave problem in the formulation (\ref{4}) in a neighborhood of $w_0$, 
where $w_0 = -\frac{2}{3} \log c_0 +
\log \left (1 + \eta q_0 f_0^\prime \right )$.
Note the size of the residual depends on the quality of quasi-solution. 
As stated in the Proposition,
the height, wave speed and crest speed parameters $(h, c, \mu)$ are all close
to $(h_0, c_0, \mu_0)$ that can be computed from the quasi-solution.
However, this does not guarantee a one to one relationship between $a_1$ 
and $\mu$ in a neighborhood of $\mu_0$.
In the following section, we determine additional conditions on quasi-solution that
ensures a one to one relationship. 
}
\end{Remark}     

\section{Enforcing the constraint (\ref{crestrelation}) for determining $a_1$}

\begin{Lemma}
\label{lema1}
The solution in Proposition \ref{prop2} satisfies the following bounds 
\be 
\| \partial_{a_1} \Phi \|_{\mathcal{E}} \le K_1 \| G \|_{\mathcal{E}}  
\ee
and
\be
\| \partial_{a_1} \Phi - G \|_{\mathcal{E}} 
\le K_1 M \| G \|_{\mathcal{E}} B_0 (1+\epsilon) \left ( \frac{26}{3} \| A \|_{\mathcal{S}} + 18\| B \|_{\mathcal{E}} 
\right )
=:K_4 \| G \|_{\mathcal{E}}
\ee
where
\be
K_1 = 
\left ( 1 - B_0 (1+\epsilon) 
\left [ \frac{26}{3} \| A \|_{\mathcal{S}} + 18 \| B \|_{\mathcal{E}} \right ] \right )^{-1}     
\ee 
\end{Lemma}
\begin{proof}
From (\ref{eqNon}), we note that $\partial_{a_1} \Phi$ satisfies
\be
\label{eqpPhi}
\partial_{a_1}  \Phi 
= \mathcal{K} \partial_{a_1} \mathcal{\tilde M} [W] + G 
\ee
where 
\be
\partial_{a_1} \mathcal{\tilde M} [W] = 
-\frac{2}{3} A \partial_{a_1} M_1 - 2 B \partial_{a_1} M_2 
\ee 
Calculation gives
\begin{multline}
\partial_{a_1} M_1 = e^{2 \Phi} \left [ \partial_{a_1} \Phi \right ]
\left ( 
2 \Re \left [ e^{W}-1-W \right ]   
+ \Re \left [ e^{W}-1 \right ] \right )
- \partial_{a_1} \Psi ~~e^{2 \Phi} ~\Im \left [e^{W}-1 \right ]     
\\
+ \partial_{a_1} \Phi \left \{ 
\left ( e^{2 \Phi} - 1 - 2 \Phi \right ) 
+ 2 (1+\Phi) \left ( e^{2 \Phi} - 1 \right ) + 4 \Phi \right \}  \ ,
\end{multline}
where $\partial_{a_1} \Psi $ is the Hilbert transform of $\partial_{a_1} \Phi $.
Therefore,
\begin{multline}
\| \partial_{a_1} M_1 \|_{\mathcal{E}}  \le 
\| \partial_{a_1} \Phi \|_{\mathcal{E}} 
\left ( 2 e^{2 \| \Phi \|_{\mathcal{E}} } \left [ e^{\| \Phi \|_{\mathcal{E}}} -1 - \| \Phi \|_{\mathcal{E}} \right ]  
+ 2 e^{2\|\Phi\|_{\mathcal{E}}} \left ( e^{\|\Phi \|_{\mathcal{E}}} - 1 \right ) \right. \\
\left. + e^{2 \| \Phi \|_{\mathcal{E}}} -1 - 
2 \|\Phi \|_{\mathcal{E}} 
+ 2 (1+\|\Phi \|_{\mathcal{E}} ) 
\left [e^{2 \|\Phi \|_{\mathcal{E}}} - 1 \right ] + 4 \| \Phi \|_{\mathcal{E}} \right )    
\end{multline}
Further, from expression for $M_2$,
\begin{multline}
\partial_{a_1} M_2 
= \partial_{a_1} \Phi \left \{ 2 e^{2 \Phi}  \Im \left ( e^{W} - 1 - W \right ) 
+ e^{2 \Phi} \Im \left (e^{W} - 1 \right ) + 2 \Psi e^{2 \Phi} \right \} \\ 
+ \partial_{a_1} \Psi \left \{ e^{2 \Phi} \Re \left (e^{W} - 1 \right ) 
+ e^{2 \Phi} - 1 \right \}  
\end{multline}
implying
\begin{multline}
\| \partial_{a_1} M_2 \|_{\mathcal{S}} \le  
\| \partial_{a_1} \Phi \| 
\left \{ 2 e^{2 \| \Phi\|_{\mathcal{E}} }  \left[ e^{\|\Phi\|_{\mathcal{E}}} - 1 - \|\Phi \|_{\mathcal{E}} \right ] 
+ e^{2 \|\Phi\|_{\mathcal{E}}} \left [e^{\|\Phi\|_{\mathcal{E}}} - 1 \right ] \right. \\
\left. + 2 \| \Phi \|_{\mathcal{E}} 
e^{2 \| \Phi\|_{\mathcal{E}}} 
+ e^{2 \|\Phi\|_{\mathcal{E}}} \left (e^{\|\Phi\|_{\mathcal{E}}} - 1 \right ) 
+ e^{2 \|\Phi\|_{\mathcal{E}}} - 1 \right \}  
\end{multline}
Therefore, when $\Phi \in \mathcal{B}$ for $B_0(1+\epsilon ) \le \frac{1}{16}$,
\be 
\| \partial_{a_1} M_1 \|_{\mathcal{E}} \le 13 \| \partial_{a_1} \Phi \|_{\mathcal{E}}  \| \Phi \|_{\mathcal{E}} 
\ ,
\| \partial_{a_2} M_2 \|_{\mathcal{E}} \le 9 \| \partial_{a_1} \Phi \|_{\mathcal{E}}  \| \Phi \|_{\mathcal{E}} 
\ee  
Therefore,
\be
\label{eq166}
\|\mathcal{K} \partial_{a_1} \mathcal{M} [ \Phi ] \|_{\mathcal{S}} 
\le M \| \partial_{a_1} \Phi \|_{\mathcal{E}} \| \Phi \|_{\mathcal{E}} 
\left ( \frac{26}{3} \| A \|_{\mathcal{S}} + 18 \| B \|_{\mathcal{E}} \right )   
\ee
The lemma readily follows
from using above bounds in 
\eqref{eqpPhi}.
\end{proof}

\begin{Proposition}
\label{propcrest}
Define $B_0$ as in (\ref{B0def}).
If in addition to 
conditions in Proposition \ref{prop2},
the following two conditions
\be
\label{eqpropcrest2}
\frac{1}{|G (\pi)|} \left [ \frac{1}{3}  \Big | \log \frac{\mu_0}{\mu} \Big |  
+ M \| R \|_{\mathcal{S}} + M \left ( 4 \| A \|_{\mathcal{S}} + 6 \| \mathcal{B} \|_{\mathcal{E}} \right ) 
B_0^2 (1 +\epsilon)^2 \right ] < \epsilon_0   
\ee
\be
\label{eqpropcrest1}
\frac{M}{|G(\pi)|} K_1 \| G \| B_0 (1 +\epsilon ) 
\left ( \frac{26}{3} \| A \|_{\mathcal{S}} + 18 \| \mathcal{B} \|_{\mathcal{E}} \right ) < 1 
\ee  
hold, then there exists unique 
$a_1 \in I = \left ( -\epsilon_0, \epsilon_0 \right )$ so that
the solution in Propositon \ref{prop2} satisfies 
(\ref{crestrelation}).
\end{Proposition}
\begin{proof}
From (\ref{a1relation}), it follows that if $a_1 \in I$, then 
\be
\Big |  \mathcal{U} [a_1] \Big |
\le 
\frac{1}{|G (\pi)|} \left [ \frac{1}{3}  \Big | \log \frac{\mu_0}{\mu} \Big |  
+ M \| R \|_{\mathcal{S}} + M \left ( 4 \| A \|_{\mathcal{S}} + 6 \| \mathcal{B} \|_{\mathcal{E}} \right ) 
B_0^2 (1 +\epsilon)^2 \right ] 
\ee
Condition (\ref{eqpropcrest2}) implies that $\mathcal{U}: I \rightarrow I$. 
Applying $\partial_{a_1}$ to (\ref{crestrelation}), and using (\ref{eq166}), Propositions \ref{prop1.1},
\ref{prop1} and Lemma \ref{lema1}, it follows by applying (\ref{eqpropcrest1}) that 
\begin{multline}
\Big | \partial_{a_1} \mathcal{U} [a_1] \Big |  
\le \Big | \frac{1}{G(\pi)} \mathcal{K} \partial_{a_1} \mathcal{M} [\Phi ] \Big | 
\\
\le \frac{1}{|G (\pi)|} M K_1 \| G \| B_0 (1+\epsilon)   
\left ( \frac{26}{3} \| A \|_{\mathcal{S}} + 18 \| \mathcal{B} \|_{\mathcal{E}} \right ) < 1 
\end{multline}
Hence $\mathcal{U}: I \rightarrow I$ is contractive, implying existence
of unique $a_1$ satisfying (\ref{crestrelation}).
\end{proof}

\section{Quasi-solution and application of Propositions \ref{prop2}, \ref{propcrest}}
\label{Results}

We describe in this section determination of quasi-solutions $(f_0, c_0)$ and
checking conditions for application of Propositions \ref{prop2} and \ref{propcrest}. Though
quasi-solutions have been obtained numerically, it has no bearing on the mathematical rigor
of Theorem \ref{Thm1} since Propositions \ref{prop2}, \ref{propcrest} 
concern
the difference $W=w-w_0$ and calculation of norms of residual $R_0$ and $R_0^\prime$ based on
$(f_0, c_0)$ are exact.

The process of obtaining quasi-solution is straight forward. As mentioned earlier,
a polynomial representation for $f_0$ is most suitable for determining exact representation
for determination of $R_0(\nu)$, $A(\nu)$ and $B(\nu)$. For that purpose, one can use
a numerical truncation of a series representation of $f$ in $\eta$ and find the
coefficients through a Newton iteration procedure involving wave speed $c$ and the series coefficient
$F_0$, $F_1$, $F_2, \cdots F_N$ for $f$ in (\ref{3.0})
by satisfying boundary condition (\ref{1.1}) at $N$ uniformly spaced out points in the upper-half
semi-circle and enforcing constraint (\ref{2.0.0.0}) for given $\mu$. Such procedures are
fairly standard and have been used routinely in the past by many investigators. 
However, such a representation for $f_0$ requires more than two hundred modes 
for 
$\| R \|_{\mathcal{S}}$ to be small enough to apply Proposition \ref{prop2} for
the values of $\mu$ quoted here.
Hence, for efficiency of 
representation and of 
presentation, a rational Pade approximant
for $f^\prime$ is found, similar to the one employed earlier by
\cite{Schwartz2};
integration and replacement of each coefficient by a ten to twelve 
digit accurate rational approximation gives rise to the quoted expressions for
${\tilde f}_0$ in the 
following subsections. 
Note that this requires specification of only upto fifty
two numerical coeefficients, compared to more than 200 otherwise. 
With well-known location
of singularities, it can be easily proved that
the truncated Taylor expansion $f_0 = \mathcal{P}_N {\tilde f}_0$
for $N = 255$ ensures that $\| f_0^\prime - {\tilde f}_0^\prime \|_{\mathcal{A}}$ is
less than $10^{-10}$ in all cases reported. 
Though $f_0$ is still a large order polynomial, we only need to list
up to fifty two 
rational numbers
for ${\tilde f}_0$ to represent $f_0$ exactly. 
A polynomial quasi-solution  
allows precise computation of all cosine
or sine series coefficients of $A$, $B$, $R_0$ and $R_0^\prime$ needed to
check conditions of Proposition \ref{prop2}. Additionally
to check conditions in Proposition \ref{propcrest}, one needs lower
bounds on $|G (\pi)|$. This is done by applying Lemma \ref{lemG0} to an
an approximate quasi-G solution $G_0 \in \mathcal{E}$ for which
$\mathcal{L} G_0 $  is small. 
We report the coefficents of $G_0$ in the appendix for
each of the three cases.
Note that a truncated rational Fourier cosine series representation for $G_0$ 
allows an exact computation of all Fourier sine series coefficients of
$\mathcal{L} G_0$ and by using these,
one estimates $\| \mathcal{L} G_0 \|_{\mathcal{S}} = \epsilon_G$. 
Since $G_0 (\pi)$ and $\| G_0 \|_{\mathcal{S}}$ are exactly known, 
positive upper and
lower bounds for
$\|G \|_{\mathcal{S}}$ and $ |G (\pi)| $ follow from Lemma \ref{lemG0}
when $G_0 (\pi) \ne 0$ for sufficiently small $\epsilon_G$.

Checking univalence condition for
$1+\eta q f_0^\prime \ne 0$ in $|\eta| \le e^{\beta}$ for suitably chosen
$\beta \ge 0$ is fairly simple, since one can determine approximate roots of
a polynomial of any order numerically. We can then express
\be 
1+\eta q f_0^\prime = \delta \prod_{j=1}^{N+2} \left ( \eta - \eta_j \right ) + z_{N+2} (\eta)      
\ee
where $\delta$ is the coefficient of $\eta^{N+2}$, $\eta_j$ 
are the numerically obtained roots approximated by rational numbers and $z_{N+2}$ is
a polynomial of degree $N+2$ with small coefficients 
which accomodates any error in the root calculations. 
In all cases reported,
$|\eta_j| > 1.09$. 
Note that though $\eta_j$ have been computed numerically, $z_{N+2}$ as a difference
of the two polynomials is known exactly. We can then check
$\|z_{N+2} \|_{\mathcal{A}}$ and prove it is small enough
for suitably chosen $\beta$, $\inf_{j} |\eta_j| > e^{\beta} \ge 1$ and
on $|\eta| = e^{\beta}$, $\Big | z_{N+2} \Big | < |\delta| \sum_{j=1}^{N+2} 
\Big | \eta - \eta_j \Big |$. By application of Rouche's theorem,
$1 + \eta q f_0^\prime 
\ne 0$ for $|\eta| \le e^{\beta}$. This also ensures analyticity of 
$w_0 = -\frac{2}{3} \log c_0 + \log (1+\eta q f_0^\prime)$
in $|\eta| \le e^{\beta}$.     
Alternately, we can show
$1+\eta q {\tilde f}_0^\prime \ne 0$ for $|\eta| \le e^{\beta}$ by
rationalizing the expression and working with the polynomial in the numerator.
The closeness of $f_0^\prime$ and ${\tilde f}_0^\prime$ implies that 
the same condition is true for $1+\eta q f_0^\prime$ from Rouche's theorem.

Recall set $S$ for which Theorem \ref{Thm1} applies:
\be
S := \left \{ \mu: \mu = \mathcal{I}_{\mu_j} \ , \mu_1 = 0.0018306, \mu_2 = 0.002 \ , \mu_3 = 0.0023
\right \} 
\ee
where $\mathcal{I}_{\mu_j}$ are sufficiently small intervals containing $\mu_j$.
We will only check in the ensuing that conditions for applicaiton of
Propositions \ref{prop2} and \ref{propcrest} apply for $\mu = \mu_j$, $j=1,2,3$. 
Since these conditions
are open set conditions; so they must hold
for a sufficient small neighborhood of $\mu=\mu_j$. The maximal sizes of
the intervals $\mathcal{I}_{\mu_j}$ which still ensures that Theorem \ref{Thm1} applies
can also be estimated if desired, though
larger size reduces the accuracy of the quasi-solution.

\subsection{Case of $\mu = \mu_1 := 0.0018306$}

In this case, we choose $c_0 = \frac{9195}{8413}$. This is close to
the empirical maximum wave speed.
wave speed and for $N=255$, we take
$f_0 = P_N {\tilde f}_0$,
where
\be 
\label{Q2.0.0.0}
{\tilde f}_0 = b_0 + \sum_{j=1}^{35} \frac{b_j}{j} \eta^{j} + 
\sum_{m=1}^8 \lambda_m \gamma_m^{-1} \log \left ( 1 + \gamma_m \eta \right ) \ ,
\ee
where 
${\bf b} = \left (b_0, b_1, \cdots , b_{35} \right ) $ is given by 
\begin{multline}
\small
{\bf b} = 
\left [ 
-{\frac {14947}{69357}},{\frac {7671}{114751}},{\frac {3587}{64227}},
{\frac {5489}{240353}},{\frac {5157}{273887}},{\frac {1747}{200565}},{
\frac {4211}{596640}},{\frac {1597}{458477}},{\frac {1055}{381241}},\right. \\
\left. {\frac {1393}{978106}},{\frac {587}{530156}},{\frac {821}{1397729}},{
\frac {524}{1174777}},{\frac {221}{912265}},{\frac {760}{4238347}},{
\frac {93}{936895}},{\frac {151}{2113416}},{\frac {213}{5300075}},\right. \\
\left. {\frac {173}{6167757}},{\frac {183}{11441683}},{\frac {61}{5654848}},{
\frac {199}{31968962}},{\frac {9}{2227843}},{\frac {74}{31411653}},{
\frac {29}{19817069}},{\frac {22}{25539231}},\right .\\
\left. {\frac {20}{39313301}},{
\frac {25}{82781219}},{\frac {7}{41644027}},{\frac {28}{278473151}},{
\frac {11}{210993463}},{\frac {7}{222912770}},{\frac {3}{200969923}},\right.\\
\left. 
{\frac {5}{552578509}},{\frac {1}{259446883}},{\frac {1}{425548468}}
\right ]
\end{multline}
$\gamma_m$, for $m=1, \cdots, 8$ given by
\be
\label{Q3.1.0.0}
\boldsymbol{\gamma} = 
\left [ 
-{\frac {279593}{312700}},-{\frac {46832}{53467}},-{\frac {29306}{
35053}},-{\frac {15231}{19853}},{\frac {34356}{45869}},{\frac {53945}{
65058}},{\frac {40025}{45693}},{\frac {289698}{322535}}
\right ]
\ee
\be
\label{Q3.1.0.0.0.0}
\boldsymbol{\lambda} = \left [ 
-{\frac {213509}{381372}},{\frac {6866}{53037}},{\frac {1248}{13703}}
,{\frac {5225}{73982}},-{\frac {1284}{177829}},-{\frac {1347}{224215}}
,-{\frac {1555}{278171}},-{\frac {42283}{7792157}}
\right ]
\ee
The height corresponding to this quasi-solution $(c_0, f_0)$ is
found to be $h_0 = 0.435905237\cdots$, while corresponding 
$\mu_0 =0.001830600034\cdots$.
For $\beta = \frac{1}{20} \log \frac{1}{\alpha}$,  
where simple Taylor series
estimates show that $\|f_0^\prime - {\tilde f}_0^\prime \|_{\mathcal{A}} \le 10^{-10}$. 
Calculations, made simple by
use of symbolic language maple, gives bounds
$ \| R_0 \|_{\mathcal{E}} \le 2.2 \times 10^{-8}$ and $\| R_0^\prime \|_{\mathcal{S}} 
\le 1.47 \times 10^{-6}$, $\| A \|_{\mathcal{S}}  \le 6.23$, 
$\| B \|_{\mathcal{E}} \le 5.34$, $\| R \|_{\mathcal{S}} \le 1.39 \times 10^{-6}$. 
With choice of $K=80$, one may check
we obtained $\gamma_{1,1} \ge 0.095$, $\gamma_{2,2} \ge 0.82 $, 
$\gamma_{1,2} \le 0.096 $, $\gamma_{2,1} \le 0.123$ and
$\gamma_f \le 1.18$, implying ${\tilde M} \le 14.3$, and $M \le 18.3$. 
Further, based on quasi-G solution $G_0$ in the appendix for this case, we
found $\epsilon_G:=\| \mathcal{L} G_0 \|_{\mathcal{S}} \le 0.011$ and therefore
from explicit calculations of $\|G_0\|_{\mathcal{E}}$ and
$G_0 (\pi)$ and using Lemma \ref{lemG0} as explained earlier, we conclude
$\|G \|_{\mathcal{E}} \le 34.7$,
$|G (\pi)| \ge 32.3 $,  
Therefore with $\epsilon_0 = 4 \times 10^{-6}$, 
$\| \Phi^{(0)} \|_{\mathcal{E}} \le 1.64\times 10^{-4} =:B_0$. 
It may be checked that 
conditions for applying Proposition \ref{prop2} hold 
when $\epsilon = \frac{3}{10}$ in which case solution $\Phi$ to the weakly nonlinear
problem exists in a ball of size
$M_E = B_0 (1+\epsilon) \le 2.2 \times 10^{-4}$ for any $a_1 \in I$. 
This is the bound $M_E$ in Theorem \ref{Thm1}.
With estimated
$\| G \|_{\mathcal{E}} \le 34.7$ and $| G (\pi) | \ge 32.3$, we also check that
conditions 
\eqref{eqpropcrest2} and \eqref{eqpropcrest1} for 
for application of proposition \ref{propcrest} for specified $\epsilon_0$ and hence
$\mathcal{U}: I \rightarrow I$ is contractive and there exists unique
$a_1$ corresponding to given $\mu$.
The constant $K_3$, estimated from a finite sum of closed form definite integrals,
satisfies $K_3 \le 5.24$ in this case.

\subsection{Case of $\mu=\mu_2 :=0.002$}

In this case, $c_0 = \frac {32419}{29662}$ and we take for $N=255$,
$f_0 = P_N {\tilde f}_0$, where
\be 
\label{Q2.0.0.0.b}
{\tilde f}_0 = b_0 + \sum_{j=1}^{35} \frac{b_j}{j} \eta^{j} + 
\sum_{m=1}^8 \lambda_m \gamma_m^{-1} 
\log \left ( 1 + \gamma_m \eta \right ) \ ,
\ee
where 
${\bf b} = \left (b_0, b_1, \cdots , b_{35} \right ) $ is given by 
\begin{multline}
\small
{\bf b} =  
\left [ -{\frac{50693}{233705}},  
{\frac {10841}{160294}},{\frac {3827}{69041}},{\frac {1833}{79169}},{
\frac {4757}{256535}},{\frac {5211}{589261}},{\frac {113}{16372}},{
\frac {1151}{325172}},
{\frac {659}{245124}}, \right. \\ 
\left. {\frac {1151}{794888}},{
\frac {445}{416281}},{\frac {443}{741637}},{\frac {629}{1469338}},{
\frac {338}{1372099}},{\frac {215}{1256464}},{\frac {335}{3319923}},{
\frac {290}{4276857}},\right. \\ 
\left. {\frac {123}{3012359}},{\frac {250}{9441639}},{
\frac {38}{2340033}},{\frac {223}{22011739}},{\frac {93}{14727322}},{
\frac {52}{13774187}},{\frac {553}{231617276}},{\frac {13}{9552156}},\right. 
\\ \left. 
{\frac {19}{21786791}},{\frac {38}{80694165}},{\frac {8}{26196413}},{
\frac {13}{83929864}},{\frac {10}{98473649}},{\frac {8}{167259223}},{
\frac {3}{94712678}},\right. \\ 
\left. {\frac {7}{513297007}},{\frac {4}{438834173}},{
\frac {1}{285147945}},{\frac {2}{845985871}}
\right ]
\end{multline}
$\gamma_m$, for $m=1, \cdots, 8$ given by
\be
\label{Q3.1.0.0.b}
\boldsymbol{\gamma} = 
\left [ 
-{\frac {53379}{59794}},-{\frac {23440}{26803}},-{\frac {7774}{9313}}
,-{\frac {22168}{28939}},{\frac {61118}{82803}},{\frac {66536}{81045}}
,{\frac {106753}{122929}},{\frac {133678}{150055}}
\right ]
\ee
\be
\label{Q3.1.0.0.0.0.b}
\boldsymbol{\lambda} = \left [ 
-{\frac {121214}{216487}},{\frac {7363}{57186}},{\frac {17911}{197419
}},{\frac {10672}{151345}},-{\frac {3205}{419337}},-{\frac {2515}{
407833}},-{\frac {6699}{1172168}},-{\frac {39545}{7133237}}
\right ]
\ee
The corresponding $h_0 = 0.4354696138\cdots$ and
$\mu_0 = 0.00199999998\cdots$.
For $\beta = \frac{1}{20} \log \frac{1}{\alpha}$,  
we use the truncated Taylor series expansion 
$f_0 = \mathcal{P}_N {\tilde f}_0$ for $N=255$
where Taylor series
estimates show that $\|f_0^\prime - {\tilde f}_0^\prime \|_{\mathcal{A}} \le 10^{-11}$. 
Calculations, made simple by
use of symbolic language maple, gives bounds
$ \| R_0 \|_{\mathcal{E}} \le 2.52 \times 10^{-8}$ and $\| R_0^\prime \|_{\mathcal{S}} 
\le 8.82 \times 10^{-6}$, $\| A \|_{\mathcal{S}}  \le 5.76$, 
$\| B \|_{\mathcal{E}} \le 4.95$, $\| R \|_{\mathcal{S}} \le 9.0 \times 10^{-6}$. 
With choice of $K=80$, one may check
we obtained $\gamma_{1,1} \ge 0.11$, $\gamma_{2,2} \ge 0.82 $, 
$\gamma_{1,2} \le 0.095 $, $\gamma_{2,1} \le 0.12$ and
$\gamma_f \le 1.15$, implying ${\tilde M} \le 12.0$, and $M \le 15.3$. 
Further, based on quasi-G solution $G_0$ in the appendix for this case, we
found $\epsilon_G:=\| \mathcal{L} G_0 \|_{\mathcal{S}} \le 0.034$ and therefore
from explicit calculations of $\|G_0\|_{\mathcal{E}}$ and
$G_0 (\pi)$ and using Lemma \ref{lemG0} as explained earlier, we conclude
$\|G \|_{\mathcal{E}} \le 29.4$,
$|G (\pi)| \ge 27.3 $,  
Therefore with $\epsilon_0 = 2 \times 10^{-6}$, 
$\| \Phi^{(0)} \|_{\mathcal{E}} \le 7.23\times 10^{-5} =:B_0$. 
It may be checked that 
conditions for applying Proposition \ref{prop2} hold 
when $\epsilon = \frac{3}{20}$ in which case solution $\Phi$ to the weakly nonlinear
problem exists in a ball of size
$M_E = B_0 (1+\epsilon) \le 8.4 \times 10^{-5}$ for any $a_1 \in I$. 
This is the bound $M_E$ in Theorem \ref{Thm1}.
With estimated
$\| G \|_{\mathcal{E}} \le 29.4$ and $| G (\pi) | \ge 27.3$, we also check that
conditions 
\eqref{eqpropcrest2} and \eqref{eqpropcrest1} for 
for application of proposition \ref{propcrest} for specified $\epsilon_0$ and hence
$\mathcal{U}: I \rightarrow I$ is contractive and there exists unique
$a_1$ corresponding to given $\mu$.
The constant $K_3$, estimated from a finite sum of closed form definite integrals,
satisfies $K_3 \le 5.20$ in this case.

\subsection{Case of $\mu = \mu_3:=0.0023$}

In this case, $c_0 = \frac {22865}{20921}$ and we take for $N=255$,
$f_0 = P_N {\tilde f}_0$, where
$P_N$ is the truncation of the Taylor series of ${\tilde f}_0$ about the origin up to and
including $\eta^N$ term, and 
\be 
\label{Q2.0.0.0.a}
{\tilde f}_0 = b_0 + \sum_{j=1}^{37} \frac{b_j}{j} \eta^{j} + 
\sum_{m=1}^6 \lambda_m \gamma_m^{-1} \log \left ( 1 + \gamma_m \eta \right ) \ ,
\ee
where 
${\bf b} = \left (b_0, b_1, \cdots , b_{37} \right ) $ is given by 
\begin{multline}
\small
{\bf b} = 
\left [ -{\frac {63307}{288588}},  
{\frac {8943}{92260}},{\frac {7094}{90177}},{\frac {5467}{140443}},{
\frac {5869}{191402}},{\frac {2551}{147538}},{\frac {3360}{254201}},{
\frac {3609}{449717}},{\frac {2951}{495374}}, \right. \\ \left.
{\frac {5189}{1363594}},{
\frac {1541}{561523}},{\frac {814}{446813}},{\frac {2661}{2084309}},{
\frac {1310}{1498987}},{\frac {465}{781442}},{\frac {409}{979023}},{
\frac {664}{2403679}},\right. \\
\left. {\frac {311}{1569858}},{\frac {431}{3391197}},{
\frac {355}{3825084}},{\frac {109}{1889706}},{\frac {117}{2735110}},{
\frac {104}{4046415}},{\frac {146}{7559227}},{\frac {83}{7419778}},{
\frac {49}{5764323}}, \right. \\
\left. {\frac {48}{10153027}},{\frac {82}{22615125}},{
\frac {25}{12976987}},{\frac {55}{36954884}},{\frac {67}{89304183}},{
\frac {35}{60078712}},{\frac {10}{36228767}},{\frac {5}{23269474}},\right.
\\
\left. 
{\frac {20}{211760067}},{\frac {7}{95298423}},{\frac {5}{170112576}},{
\frac {7}{308113800}}
\right ]
\end{multline}
$\gamma_m$, for $m=1, \cdots, 6$ given by
\be
\label{Q3.1.0.0.a}
\boldsymbol{\gamma} = 
\left [ 
-{\frac {155593}{174744}},-{\frac {119606}{138125}},-{\frac {21931}{
27114}},{\frac {64810}{84503}},{\frac {16039}{18975}},{\frac {154314}{
175871}}
\right ]
\ee
\be
\label{Q3.1.0.0.0.0.a}
\boldsymbol{\lambda} = \left [ 
-{\frac {184732}{341273}},{\frac {10272}{73979}},{\frac {15677}{
156817}},-{\frac {2844}{284831}},-{\frac {1136}{156707}},-{\frac {
48187}{7142266}}
\right ]
\ee
The corresponding $h_0 =0.4347167189\cdots$ and
$\mu_0 = 0.00230000015\cdots$.
For $\beta = \frac{1}{20} \log \frac{1}{\alpha}$, a 
truncated Taylor series expansion of $f_0 = \mathcal{P}_N {\tilde f}_0$ to a degree of $N=255$
gives rise to $\| f_0^\prime - {\tilde f}_0^\prime \|_{\mathcal{A}} \le 10^{-11}$. For
$f_0$ as above,
$ \| R_0 \|_{\mathcal{E}} \le 1.065 \times 10^{-7}$ and $\| R_0^\prime \|_{\mathcal{S}} 
\le 5.33 \times 10^{-6}$, $\| A \|_{\mathcal{S}}  \le 5.10 $, 
$\| B \|_{\mathcal{E}} \le 4.40$. 
Based on this, $ \| R \|_{\mathcal{S}} \le  
5.1 \times 10^{-6} $ 
With choice of $K=80$,
we obtained $\gamma_{1,1} \ge 0.137$, $\gamma_{2,2} \ge 0.84 $, 
$\gamma_{1,2} \le 0.091 $, $\gamma_{2,1} \le 0.11$ and
$\gamma_f \le 1.10$, implying ${\tilde M} \le 8.98$, and $M \le 11.3$. Based on
the quasi-G solution $G_0$ in the appendix for this case, we calculated the
Fourier sine coefficient of
$\mathcal{L} G_0$ and estimated
$\| \mathcal{L} [G_0] \|_{\mathcal{S}} \le 0.0005=:\epsilon_G$. 
Using it in Lemma \ref{lemG0} with
explicit calculation of $\|G_0\|_{\mathcal{E}} $ and $|G_0(\pi)|$, we get  
the bounds $\|G \|_{\mathcal{E}} \le 23.21$,
$|G (\pi)| \ge 21.87 $,  
and therefore with $\epsilon_0 = 4 \times 10^{-6}$, 
$B_0 = \| \Phi^{(0)} \|_{\mathcal{E}} \le 1.5 \times 10^{-4}$. 
Contraction mapping
argument follows for $\epsilon = \frac{1}{10}$ giving rise to a ball size 
$B_0 (1+\epsilon) \le 1.65 \times 10^{-4}$ for any $a_1 \in I$. 
where solution exists for
$\Phi $ to the weakly nonlinear problem. This is bound $M_E$ in Theorem \ref{Thm1}.
To prove that there exists $a_1 \in I$ satisfying constraint (\ref{crestrelation})
we checked that 
both conditions 
\eqref{eqpropcrest2} and \eqref{eqpropcrest1} for 
contraction of $\mathcal{U}: I \rightarrow I$.
were valid. 
The constant $K_3$, estimated from a finite sum of closed form definite integrals,
satisfies $K_3 \le 5.14$ in this case.

\section{Discussion}

We have shown how, through construction of quasi-solutions $(f_0, c_0)$ obtained
through numerical calculations, one can rigorously
and constructively prove existence of water wave solution by turning the
strongly nonlinear problem into a weakly nonlinear analysis.
Thus far, we have only demonstrated this for a small set of $\mu$ in the range
$\left ( 0, \frac{1}{3} \right )$.

The quasi-solution can be determined also with explicit dependence on $\mu$ over 
suitably small intervals of $\mu$ by using small order polynomials in $\mu$ for
coefficients of the rational approximant ${\tilde f}_0$. Proving the residuals
$R_0$ and $R_0^\prime$ is also not difficult since the cosine or sine series
involving $\cos (n \nu)$ or $\sin (n \nu)$
now involve polynomials in $\mu$, which can be expressed as
a Chebyshev basis in scaled $\mu$ variable. An $l^1$ estimate of these Chebyshev
coefficients gives the maximal value of the coefficient 
of $\cos (n \nu)$ or $\sin (n\nu)$.

However, the proof thus far is manageable
(with help of symbolic manipulation language MAPLE) 
for $\mu $ relatively large, which corresponds to modest $h$, where
Stokes original expansion works just as well. Hence we have limited presentations for
small intervals around isolated values of $\mu$. The corresponding wave heights
are somewhat smaller than the critical.
When $\mu \rightarrow 0$, the accuracy needs for quasi-solution becomes
more stringent since the bound $M$ in our method deteriorates. The present rational
function approximation gets taxed to the limit when $\mu$ becomes very small. 
For more efficiency in these cases,
it is better to incorporate local behavior near the crest as was done earlier
in numerical computations \cite{Williams}. Unfortunately, the simple
emprical approximation due to Longuet-Higgins \cite{Longuet-Higgins2} is not accurate
enough to be controlled rigorously. One also needs a closer examination of
of the Nekrasov integral formulation which we suspect will work better for
higher waves than 
the simple minded, though general, series
method employed here.

Nonetheless, what is also interesting in this approach 
is that detailed features of the solution
that are difficult to prove in non-constructive methods
can be obtained with relative ease.
For instance, a crucial
role in the stability of periodic water waves 
is played by the empricial fact that wave speed $c$ 
goes through a maximum close to 
$\mu =\mu_1 = 1.8306\times 10^{-3}$. This can be confirmed
in the following manner. We take two values on either side of $\mu_1$ and compute
$\partial_{a_1} W (\alpha)$, which upto nonlinear correction is given by
$G(\alpha)$. Through a more
accurate representation of quasi-G solution $G_0$ than provided here, 
it can be proved that
$\partial_{a_1} W (\alpha) $ which determines $\frac{d c}{da_1}$ changes in some
interval around $\mu=\mu_1$. Control of the the lower bound of second derivative is also needed
to prove that there is only one such maximum of $c$ in some interval.
In this context, it is interesting to note
that even for our relatively inaccurate 
quasi-G solution $G_0$ for $\mu = \mu_1$, we find $G_0 (\alpha) = 9.97\cdots \times 10^{-6}$,
which is signicantly smaller than 
$4.86\cdots \times 10^{-3}$.
at $\mu=\mu_2$, suggesting that $G(\alpha)$ does change sign for some $\mu$ close to
$\mu_1$.

\section{Acknowledgment}

This research was partially supported by the National Science Foundation (NSF-DMS-1108794).
The author is also deeply indebted to Ovidiu Costin for many collaborations
in other contexts involving the quasi-solution idea. Furthermore, the author
wishes to thank Jerry Bona for his encouragement, 
hospitality and support during the author's sabbatical stay at
UIC.

\section{Appendix}

Here we simply present the quasi-G solution $G_0$ found numerically for different
$\mu$ with coefficients approximated by rationals to 8 digits. It is clear that
if $|G_0 (-1)|$ is sufficiently large, which it is in all the cases presented,
$\mathcal{L} G_0 $ need not be too small to check Proposition \ref{propcrest}.
In all cases, the quasi-solution is taken to be in the form 
\be
\label{eqG0}
G_0 = \sum_{j=0}^{126} g_j \cos (j \nu) 
\ee
\subsection{$G_0$ Fourier cosine coefficents $\mu = \mu_1 := 0.0018306$}
${\bf g} =  \left ( g_0, g_1, g_2, \cdots g_{126} \right ) $ is given by:  
\begin{multline}
\small
\left [ 
{\frac {997}{19434}},1,-{\frac {7874}{3587}},{\frac {3214}{1445}},-{
\frac {6536}{2921}},{\frac {12795}{6056}},-{\frac {6399}{3214}},{
\frac {6783}{3683}},-{\frac {1337}{786}},{\frac {750}{481}},-{\frac {
2585}{1811}},{\frac {3479}{2672}},\right. \\
\left. -{\frac {2210}{1863}},{\frac {1081}{
1002}},-{\frac {2909}{2968}},{\frac {6520}{7329}},-{\frac {4591}{5691}
},{\frac {2598}{3553}},-{\frac {2931}{4426}},{\frac {3030}{5053}},-{
\frac {2033}{3747}},{\frac {1974}{4021}},-{\frac {1689}{3805}},{\frac 
{1146}{2855}},\right .\\ 
\left.-{\frac {859}{2368}},{\frac {3085}{9409}},-{\frac {598}{
2019}},{\frac {1136}{4245}},-{\frac {1558}{6447}},{\frac {1937}{8874}}
,-{\frac {617}{3131}},{\frac {773}{4344}},-{\frac {5259}{32743}},{
\frac {724}{4993}},-{\frac {1042}{7963}},{\frac {1413}{11963}},\right. \\
\left. -{
\frac {5905}{55408}},{\frac {1302}{13537}},-{\frac {692}{7975}},{
\frac {823}{10511}},-{\frac {1574}{22285}},{\frac {341}{5351}},-{
\frac {898}{15623}},{\frac {457}{8813}},-{\frac {845}{18068}},{\frac {
695}{16474}},-{\frac {637}{16743}},\right.\\
\left. {\frac {578}{16843}},-{\frac {413}{
13346}},{\frac {146}{5231}},-{\frac {113}{4490}},{\frac {257}{11323}},
-{\frac {169}{8258}},{\frac {233}{12625}},-{\frac {851}{51143}},{
\frac {572}{38121}},-{\frac {279}{20624}},{\frac {428}{35087}},\right.\\
\left. -{
\frac {699}{63562}},{\frac {133}{13413}},-{\frac {175}{19577}},{\frac 
{257}{31887}},-{\frac {263}{36198}},{\frac {73}{11144}},-{\frac {687}{
116342}},{\frac {285}{53534}},-{\frac {181}{37717}},{\frac {219}{50620
}},\right. \\
\left. -{\frac {279}{71543}},{\frac {106}{30151}},-{\frac {473}{149263}},{
\frac {109}{38156}},-{\frac {41}{15923}},{\frac {118}{50837}},-{\frac 
{111}{53056}},{\frac {125}{66281}},-{\frac {109}{64125}},{\frac {394}{
257143}},\right .\\
\left. -{\frac {121}{87618}},{\frac {208}{167093}},-{\frac {19}{
16935}},{\frac {70}{69219}},-{\frac {99}{108619}},{\frac {17}{20693}},
-{\frac {49}{66179}},{\frac {194}{290695}},-{\frac {75}{124696}},{
\frac {41}{75630}},\right . 
\\
\left. -{\frac {50}{102339}},{\frac {19}{43147}},-{\frac {
23}{57955}},{\frac {93}{260003}},-{\frac {57}{176824}},{\frac {104}{
357963}},-{\frac {37}{141313}},{\frac {31}{131367}},-{\frac {38}{
178685}},{\frac {302}{1575655}},\right .\\
\left. -{\frac {13}{75263}},{\frac {38}{
244105}},-{\frac {32}{228103}},
{\frac {20}{158187}},-{\frac {47}{
412506}},{\frac {23}{223988}},-{\frac {52}{561947}},{\frac {13}{155885
}},-{\frac {11}{146370}},
{\frac {32}{472479}}, \right .\\ 
\left. -{\frac {28}{458767}},{
\frac {15}{272711}},-{\frac {9}{181576}},{\frac {17}{380580}},-{\frac 
{37}{919192}},{\frac {58}{1598885}},-{\frac {27}{825968}},{\frac {14}{
475243}},-{\frac {7}{263693}},{\frac {19}{794229}},\right. \\
\left. -{\frac {31}{
1438034}},{\frac {10}{514757}},-{\frac {11}{628366}},{\frac {7}{443727
}},-{\frac {4}{281383}},{\frac {11}{858680}},-{\frac {25}{2165713}},{
\frac {13}{1249706}},-{\frac {7}{746770}}
\right ]
\end{multline}

\subsection{$G_0$ Fourier cosine coefficients for  $\mu = \mu_2:=0.002$}

In this case, ${\bf g} = \left ( g_0, g_1, \cdots, g_{126} \right )$ given by
\begin{multline} 
{\bf g} = 
\left [ 
{\frac {975}{83047}},1,-{\frac {2174}{1053}},{\frac {6185}{3017}},-{
\frac {4355}{2136}},{\frac {14321}{7532}},-{\frac {1807}{1018}},{
\frac {7847}{4829}},-{\frac {3627}{2438}},{\frac {2457}{1819}},-{
\frac {3787}{3089}},\right.\\
\left.{\frac {9427}{8507}},-{\frac {4444}{4439}},{\frac 
{1397}{1548}},-{\frac {2770}{3407}},{\frac {2254}{3081}},-{\frac {
10127}{15393}},{\frac {9302}{15733}},-{\frac {5913}{11135}},{\frac {
1027}{2154}},-{\frac {1936}{4525}},{\frac {719}{1873}},\right. \\
\left. -{\frac {768}{
2231}},{\frac {20423}{66159}},-{\frac {11843}{42804}},{\frac {466}{
1879}},-{\frac {843}{3794}},{\frac {20679}{103865}},-{\frac {3268}{
18327}},{\frac {1486}{9303}},-{\frac {1340}{9369}},{\frac {793}{6191}}
,-{\frac {319}{2782}},\right. \\ 
\left.{\frac {5586}{54407}},-{\frac {1289}{14027}},{
\frac {2848}{34619}},-{\frac {1552}{21081}},{\frac {886}{13445}},-{
\frac {123}{2086}},{\frac {572}{10839}},-{\frac {289}{6121}},{\frac {
587}{13893}},-{\frac {471}{12461}},{\frac {298}{8811}},\right. \\ 
\left.-{\frac {242}{
7999}},{\frac {688}{25417}},-{\frac {1083}{44731}},{\frac {977}{45105}
},-{\frac {799}{41243}},{\frac {445}{25677}},-{\frac {272}{17549}},{
\frac {119}{8583}},-{\frac {293}{23631}},{\frac {371}{33452}},-{\frac 
{501}{50516}},\right.\\
\left. {\frac {56}{6313}},-{\frac {173}{21810}},{\frac {396}{
55819}},-{\frac {247}{38937}},{\frac {395}{69624}},-{\frac {213}{41989
}},{\frac {163}{35930}},-{\frac {58}{14299}},{\frac {113}{31152}},-{
\frac {180}{55501}},{\frac {317}{109303}},\right. \\
\left. -{\frac {56}{21597}},{\frac 
{95}{40972}},-{\frac {553}{266767}},{\frac {118}{63659}},-{\frac {39}{
23534}},{\frac {50}{33743}},-{\frac {187}{141162}},{\frac {93}{78515}}
,-{\frac {229}{216260}},{\frac {46}{48585}},\right. \\ 
\left. -{\frac {53}{62618}},{
\frac {59}{77963}},-{\frac {79}{116775}},{\frac {51}{84317}},-{\frac {
65}{120213}},{\frac {169}{349587}},-{\frac {40}{92561}},{\frac {207}{
535768}},-{\frac {64}{185307}},{\frac {83}{268804}},\right. \\
\left. -{\frac {417}{
1510792}},{\frac {33}{133732}},-{\frac {24}{108805}},{\frac {51}{
258623}},-{\frac {29}{164519}},{\frac {25}{158644}},-{\frac {5}{35496}
},{\frac {29}{230292}},-{\frac {20}{177681}},{\frac {131}{1301845}},\right. \\
\left. -{
\frac {35}{389126}},{\frac {22}{273607}},-{\frac {10}{139137}},{\frac 
{29}{451365}},-{\frac {13}{226368}},{\frac {15}{292183}},-{\frac {42}{
915289}},{\frac {19}{463190}},-{\frac {76}{2072849}},\right . \\ 
\left. {\frac {10}{
305109}},-{\frac {17}{580304}},{\frac {35}{1336533}},-{\frac {9}{
384511}},{\frac {5}{238971}},-{\frac {17}{909037}},{\frac {13}{777658}
},-{\frac {8}{535421}},{\frac {13}{973341}},\right. \\ 
\left. -{\frac {5}{418846}},{
\frac {10}{937139}},-{\frac {6}{629101}},{\frac {5}{586492}},-{\frac {
27}{3543425}},{\frac {8}{1174561}},-{\frac {6}{985613}},{\frac {8}{
1470193}},-{\frac {5}{1028078}}, \right. \\
\left. {\frac {10}{2300319}},-{\frac {3}{
772117}},{\frac {5}{1439682}},-{\frac {2}{644321}}
\right ] 
\end{multline}

\subsection{Fourier cosine coefficients for $G_0$ in case $\mu = \mu_3:= 0.0023$}

In this case, 
${\bf g} =  \left ( g_0, g_1, g_2, \cdots g_{126} \right ) $ is given by  
\begin{multline}
\small
\left [ 
-{\frac {715}{18432}},1,-{\frac {4723}{2490}},{\frac {5442}{2981}},-{
\frac {4351}{2441}},{\frac {1839}{1129}},-{\frac {1412}{943}},{\frac {
1459}{1083}},-{\frac {39833}{32774}},{\frac {4286}{3947}},-{\frac {
5351}{5508}},\right. \\ 
\left. {\frac {2569}{2971}},-{\frac {1109}{1440}},{\frac {1841}{
2692}},-{\frac {3146}{5179}},{\frac {3741}{6946}},-{\frac {2436}{5101}
},{\frac {3291}{7781}},-{\frac {3812}{10177}},{\frac {1671}{5041}},-{
\frac {2186}{7453}},{\frac {1475}{5686}},\right. \\ 
\left. -{\frac {954}{4159}},{\frac {
3809}{18784}},-{\frac {883}{4927}},{\frac {1343}{8480}},-{\frac {557}{
3981}},{\frac {155}{1254}},-{\frac {517}{4736}},{\frac {683}{7084}},-{
\frac {743}{8728}},{\frac {1389}{18478}},-{\frac {357}{5380}},\right. \\
\left. {\frac {
447}{7630}},-{\frac {9247}{178837}},{\frac {7591}{166312}},-{\frac {
1419}{35230}},{\frac {1310}{36849}},-{\frac {333}{10616}},{\frac {494}
{17845}},-{\frac {309}{12652}},{\frac {1118}{51875}},-{\frac {520}{
27351}},\right. \\ 
\left. {\frac {493}{29388}},-{\frac {416}{28113}},{\frac {200}{15319}
},-{\frac {289}{25097}},{\frac {231}{22738}},-{\frac {163}{18192}},{
\frac {721}{91216}},-{\frac {11}{1578}},{\frac {538}{87491}},-{\frac {
215}{39648}},\right. \\ 
\left. {\frac {133}{27805}},-{\frac {214}{50735}},{\frac {272}{
73109}},-{\frac {50}{15241}},{\frac {57}{19699}},-{\frac {305}{119544}
},{\frac {149}{66215}},-{\frac {125}{63002}},{\frac {77}{44004}},-{
\frac {35}{22686}}, \right. \\
\left. {\frac {155}{113918}},-{\frac {258}{215071}},{
\frac {97}{91689}},-{\frac {63}{67546}},{\frac {292}{355007}},-{\frac 
{175}{241333}},{\frac {69}{107903}},-{\frac {57}{101110}},{\frac {65}{
130752}},\right. \\ 
\left. -{\frac {110}{250999}},{\frac {179}{463188}},-{\frac {3}{8806
}},{\frac {40}{133153}},-{\frac {62}{234123}},{\frac {15}{64237}},-{
\frac {65}{315774}},{\frac {353}{1944853}},-{\frac {25}{156253}},\right. \\
\left. 
{\frac {86}{609597}},-{\frac {25}{201033}},{\frac {29}{264477}},-{
\frac {31}{320731}},{\frac {20}{234681}},-{\frac {359}{4779007}},{
\frac {16}{241567}},-{\frac {17}{291184}},{\frac {64}{1243307}},\right.
\\
\left. -{
\frac {17}{374674}},{\frac {32}{799907}},-{\frac {14}{397035}},{\frac 
{53}{1704772}},-{\frac {15}{547391}},{\frac {73}{3021506}},-{\frac {49
}{2301003}},{\frac {4}{213049}},-{\frac {20}{1208577}},\right.
\\ \left. {\frac {12}{
822485}},-{\frac {19}{1477506}},{\frac {10}{882027}},-{\frac {21}{
2101522}},{\frac {11}{1248582}},-{\frac {7}{901486}},{\frac {13}{
1898967}},-{\frac {35}{5800724}},{\frac {6}{1127929}},\right. \\
\left. -{\frac {7}{
1493042}},{\frac {1}{241932}},-{\frac {5}{1372493}},{\frac {5}{1556798
}},-{\frac {7}{2472921}},{\frac {8}{3205731}},-{\frac {7}{3182644}},{
\frac {6}{3094345}},\right. \\
\left. -{\frac {4}{2340633}},{\frac {1}{663748}},-{\frac 
{1}{753116}},{\frac {4}{3417071}},-{\frac {3}{2907881}},{\frac {1}{
1099486}},-{\frac {4}{4990151}},{\frac {3}{4245329}},\right. 
\\ 
\left. -{\frac {1}{
1605671}},{\frac {2}{3642713}},-{\frac {2}{4133271}}
\right ]
\end{multline}

\vfill \eject
\end{document}